\documentclass[11pt,letterpaper]{article}
\usepackage{times,fullpage,natbib,xspace}
\usepackage{mysetup,my-acm} 

\usepackage[ruled]{algorithm2e}

\renewcommand{\refeq}[1]{(\ref{#1})}
\renewcommand{\eqref}[1]{Equation~(\ref{#1})}

\newcommand{\mA}{\mathcal{A}}
\newcommand{\mD}{\mathcal{D}}
\newcommand{\mF}{\mathcal{F}}
\newcommand{\mM}{\mathcal{M}}
\newcommand{\mP} {\ensuremath{\mathcal{P}}} 

\newcommand{\Skip}{\ensuremath{\mathtt{skip}}\xspace}
\newcommand{\CMON}{\ensuremath{\mathtt{CMON}}\xspace}
\newcommand{\WMON}{\ensuremath{\mathtt{WMON}}\xspace}
\newcommand{\MIDR}{\ensuremath{\mathtt{MIDR}}\xspace}

\newcommand{\tA}{\widetilde{\mA}}
\newcommand{\tP}{\widetilde{\mP}}

\newcommand{\bid}{\mathbf{b}}
\newcommand{\outcomes}{\mathcal{O}}
\newcommand{\type}{\mathbf{x}}
\newcommand{\types}{\mathcal{T}}
\newcommand{\GenericMech}{{\tt AllocToMech}}
\newcommand{\noPositiveTransfers}{no-positive-transfers}

\newcommand{\Norm}[2][1]{ \left\| #2 \right\|_{#1} }
\newcommand{\bMin}{b_{\mathtt{min}}}
\newcommand{\bMax}{b_{\mathtt{max}}}
\newcommand{\AMax}{A_{\mathtt{max}}}

\newcommand{\bt}{\tilde{b}}

\newcommand{\UCB}{\ensuremath{\mathtt{UCB1}}\xspace}
\newcommand{\Rand}{\mathtt{RND}} 
\newcommand{\MosheMech}{\mathtt{SubSample}}

\title{Multi-parameter Mechanisms with Implicit Payment Computation%
\footnote{This is the full version of a conference paper in \emph{ACM EC 2013}.\newline
The multi-parameter transformation from Section~\ref{sec:transform} has  appeared in~\citep{Transform-ec10-jacm}, the full version of~\citep{Transform-ec10}, but it is not a part of the conference version of~\citep{Transform-ec10}. This result is an integral part of this manuscript and the corresponding conference submission.}}

\author{Moshe Babaioff%
    \thanks{Microsoft Research Silicon Valley, Mountain View, CA 94043, USA.
    Email: {\tt moshe@microsoft.com}.}
\and
Robert D. Kleinberg
    \thanks{Computer Science Department, Cornell University, Ithaca, NY 14853, USA. Email: {\tt rdk@cs.cornell.edu}. 
    \newline Parts of this research have been done while R.\ Kleinberg was a Consulting Researcher at Microsoft Research Silicon Valley, and R.\ Kleinberg was partially supported by NSF Awards CCF-0643934, AF-0910940, and IIS-0905467, and by a Microsoft Research New Faculty Fellowship.}
\and
Aleksandrs Slivkins%
    \thanks{Microsoft Research Silicon Valley, Mountain View, CA 94043, USA.
    Email: {\tt slivkins@microsoft.com}.
    \newline Parts of this research have been done while A. Slivkins was a visiting researcher at Microsoft Research New York. }
}

\date{First version: February 2013 \\This version: May 2013}

\begin{document}

\maketitle

\begin{abstract}
In this paper we show that payment computation essentially does not present any obstacle in designing truthful mechanisms, even for multi-parameter domains, and even when we can only call the allocation rule once.
We present a general reduction that takes any  allocation rule which satisfies ``cyclic monotonicity'' (a known necessary and sufficient condition for truthfulness) and converts it to a truthful mechanism using a single call to the allocation rule, with arbitrarily small loss to the expected social welfare.

A prominent example for a multi-parameter setting in which an allocation rule can only be called once arises in sponsored search auctions. These are multi-parameter domains when each advertiser has multiple possible ads he may display, each with a different value per click. Moreover, the mechanism typically does not have complete knowledge of the click-realization or the click-through rates (CTRs); it can only call the allocation rule a single time and observe the click information for ads that were presented. 
On the negative side, we show that an allocation that is truthful for any realization essentially cannot depend on the bids, and hence cannot do better than random selection for one agent.
We then consider a relaxed requirement of truthfulness, only in expectation over the CTRs.
Even for that relaxed version, making any progress is challenging as standard techniques for construction of truthful mechanisms (as using VCG or an MIDR allocation rule) cannot be used in this setting.
We design an allocation rule with non-trivial performance and directly prove it is cyclic-monotone, and thus it can be used to create a truthful mechanism using our general reduction.
\end{abstract}

\xhdr{ACM categories.}
.

\xhdr{Keywords.} algorithmic mechanism design, multi-parameter mechanisms, multi-armed bandits.

\newpage

\section{Introduction}
\label{sec:intro}

In this paper we show that payment computation essentially does not present any obstacle in designing truthful mechanisms, even for multi-parameter domains, and even when we can only call the allocation rule once. This extends the result of~\citep{Transform-ec10} for single parameter domains to multi-parameter domains. We present a general reduction that takes any  allocation rule which satisfies ``cyclic monotonicity'' (a known necessary and sufficient condition for truthfulness) and convert it to a truthful mechanism using a single call to the allocation rule, with arbitrarily small loss to the expected social welfare.
The mechanism 
does not compute the payments explicitly but rather charges random payments having the right expectation.


Such a reduction is particularly attractive as it can handle multi-parameter settings where it is impossible to decouple the computation of the allocation from the actual execution of the allocation. In such situations, the entire mechanism --- including the payment computation --- can only execute a single call to the allocation rule. We call this the ``no-simulation'' constraint; it
can arise when a mechanism interacts with the environment, and the information revealed by the environment depends on the choices made by the allocation rule. The no-simulation constraint is a significant hurdle because the existing approaches to payment computation require multiple calls to the allocation rule, with different vectors of bids.

Sponsored search auctions supply a prominent example of a multi-parameter
setting with the no-simulation constraint.
In this setting each advertiser has multiple possible ads he is interested in displaying, each with a different value per click, and the mechanism does not have complete knowledge of the click-realization or the click-through rates (CTRs).  Instead, it can only allocate ad impressions and observe the click information for ads that were presented. The no-simulation constraint also arises in other contexts, such as packet routing~\citep{Parkes-netecon12}.

We note that our reduction --- the multi-parameter transformation --- has other uses beyond settings with the no-simulation constraint. For example, it can also be used to speed up the computation of payments in most multi-parameter mechanisms. Indeed, it has already been used for this purpose by two recent papers.
\citet{Jain-sagt11} used it to speed up the payment computation for a mechanism that allocates batch jobs in a cloud system. \citet{Huang-ExpMech12} used it to compute payments for their privacy-preserving procurement auction for spanning trees, which is based on the well-known ``exponential privacy mechanism" from prior work~\citep{McSherry-Talwar-focs07}.

\xhdr{Sponsored search mechanisms with unknown CTRs.}
In the remainder of the paper we focus
on the problem of designing truthful mechanisms for an
archetypical multi-parameter setting with the no-simulation
constraint: sponsored search auctions with unknown click-through
rates (CTRs).
The difficulty in designing such allocation rules
stems from the fact that the welfare of a given allocation
depends on clicks of the allocated ads, which are unknown
to the bidders and to the mechanism.
This prevents us from using the VCG mechanism since it depend on choosing a welfare-maximizing allocation.
Yet, it is possible that welfare can at least be approximated.

We focus on a simple single-shot ad auction in which the allocation rule unfolds over time (and the CTRs are not known). As such, we contribute to a growing literature on ad auctions that unfold over time, as they do in practice. The non-strategic version of our model is a well-understood variant of the \emph{multi-armed bandit} problem.

Mechanisms that are truthful for every realization of the clicks would be most attractive, as the strategic behavior in such mechanisms would not depend on the agents' beliefs about the process generating the clicks --- for example, the belief that clicks for each ad are i.i.d.\ from a fixed distribution.
Such mechanisms were constructed in~\cite{MechMAB-ec09, Transform-ec10} for the single-parameter version of the problem.
Unfortunately, the multi-parameter setting is much harder.
In the setting of sponsored search with multiple ads per
bidder and unknown CTRs, we show that if the mechanism is
required to be truthful for every realization of the clicks, then it must be
a trivial mechanism that outputs a fixed allocation (or
distribution over allocations) with no dependence on the bids.
\OMIT{
That is, we show that the allocation essentially cannot depend
on the bids\footnote{More precisely, we show in Lemma~\ref{lm:ExPostWMON-rand}
below that the set of all distributions over allocations that the mechanism
might assign to a bidder
with minimum bid $B$ has diameter tending to zero as $B$ tends to
infinity. \BKnote{The wording of this footnote is unfortunately pretty
hard to understand.}} and thus cannot improve the approximation
guarantee of the na\"ive random allocation for a single agent.
}

In light of this negative result we consider a weaker notion of truthfulness. Assume that clicks are stochastic (meaning that each ad has a CTR, and clicks are independent Bernoulli trials with the specified click probabilities) but the CTRs are not known. The mechanism is required to be truthful for every vector of CTRs; we call mechanisms with this property {\em stochastically truthful}.
The VCG mechanism still cannot be used as we cannot maximize the expected welfare without knowing the CTRs.
An alternative is to use a maximal-in-distributional-range (\MIDR) allocation rule combined with VCG-based payment rule, but we show that for a natural family of \MIDR allocation rules (in which the set of distributions the rule optimizes over is independent of the CTRs) the performance of such rules is no better than randomly selecting an ad to present.

There are a few examples in the literature of non-VCG-based
truthful multi-parameter mechanisms in which bidders
freely choose an option from a hand-crafted menu of allocations and prices,
e.g.~\citep{BGN03,DNS06,DN11}, but this technique similarly fails in
our setting because the bidders do not have a dominant strategy
for choosing from such a menu when they do not know their own CTRs.

Given all these negative results we turn to our multi-parameter transformation which reduces the problem of designing truthful randomized mechanisms to the (seemingly simpler) problem of designing cyclically monotone (\CMON) allocation rules.
In contrast to the negative result for truthfulness for every realization, we
directly craft an allocation rule that
satisfies stochastic \CMON; to our knowledge, the only previous paper to
successful apply this approach is~\citep{LS07}.
Using the transformation we construct a stochastically truthful mechanism that
outperforms the na\"ive random allocation for a single agent, when the difference in value-per-impression of his ads is sufficiently large.
While this is clearly just a small step, it proves to be rather challenging,
and relies heavily on the multi-parameter transformation described above.



\OMIT{
Our goal now is to design a stochastically truthful mechanism that performs better than the na\"ive random allocation for a single agent, clearly a modest goal yet rather challenging. The standard approach for designing such a truthful multi-parameter mechanisms would be to design an MIDR mechanism, but such mechanisms perform poorly in our setting (IS IT SO?).

Thus we take on the challenge of designing a stochastic ''cycle monotone'' allocation rule that is not MIDR (or essentially single parameter).
Indeed, we design a cyclic-monotone allocation rule whose expected social welfare exceeds that of the random selection. As far as we know this is the first result of that kind.

FINISH!

-----------------------

BELOW IS THE OLD abstract by Alex, I suggest removing it:

We study multi-parameter mechanisms in settings where one cannot simulate the allocation rule, in the sense that the mechanism, including the payment computation, can only execute a single call to the allocation rule. The ``no-simulation'' property can arise when a mechanism interacts with the environment, and the information revealed by the environment depends on the choices made by the allocation rule. This is a significant  limitation because the existing approaches to payment computation require multiple calls to the allocation rule, with different vectors of bids.

We have two obstacles to overcome: design an allocation rule which satisfies ``cyclic monotonicity'' (a known necessary condition), and then, given such allocation rule, compute payments that induce incentive-compatibility. Our main result is that we fully resolve the payment computation obstacle. We give a general way to transform any cyclic-monotone allocation rule to a truthful mechanism which only executes a single call to the allocation rule, with arbitrarily small loss to the expected social welfare. This result -- the multi-parameter transformation -- can also be used to speed up the payment computation in settings without the ''no-simulation'' property (and in fact it has already been used in
MOSHE: CITE and remove  ''two papers by other authors'').

The design of cyclic-monotone allocation rules is generally known to be a huge challenge in the literature on mechanism design. In particular, no such allocation rules are known for multi-parameter settings with ''no-simulation'' property. In this paper, we ask a very basic, yet challenging question: for a given well-motivated setting with a ''no-simulation'' property, do there exist non-trivial cyclic-monotone allocation rules? More specifically, we are looking for allocation rules that depend on all submitted bids, and improve over a trivial allocation rule that chooses an outcome uniformly at random among all feasible outcomes.

We consider a natural multi-parameter extension to ''MAB mechanisms'', a paradigmatic mechanism design problem with a ''no-simulation'' property that has been studied in prior work. We obtain two results. First, we find that our question is not vacuous: it has a negative answer for a stricter version of the problem where one looks for truthfulness and cyclic-monotonicity for every realization of randomness in the environment. More precisely, for this version every cyclic-monotone mechanism cannot depend on the bids. On the other hand, the question is not hopeless: we give a positive answer for a more relaxed version of the problem where truthfulness and cyclic-monotonicity are in expectation over the environment. Namely, we design a cyclic-monotone allocation rule whose expected social welfare exceeds that of the random selection.
}

\xhdr{Related work.}
Our earlier paper~\citep{Transform-ec10} considers the limited case of single parameter domains.
It introduced the technique
of designing black-box transformations that perform implicit
payment computation while evaluating a given monotone allocation function
only once. The same paper introduced monotone allocation rules
with strong welfare guarantees for
sponsored search auctions with unknown CTRs, by
modifying multi-armed bandit algorithms to achieve
the requisite monotonicity property.
As all the results in our earlier paper are limited to single-parameter settings, they only apply to sponsored search when each advertiser has only {\em one} ad to display.
In the present paper, we show that the black-box
transformation extends readily from single-parameter
to multi-parameter settings, whereas extending
the results on sponsored search to multi-parameter
settings is much more delicate, and in some cases
(i.e.\ for the strongest notion of truthfulness)
outright impossible.

\citet{SingleCall-ec12} extended the results
of~\citep{Transform-ec10} to multi-parameter
domains under some limitations. Their work
provides a
black-box transformation that allows implicit payment
computation when the allocation function is
maximal-in-distributional-range (\MIDR).
While the \MIDR property is the most widely used method
for achieving truthfulness in multi-parameter settings,
it is not a necessary condition for truthfulness.
In fact several papers
(including this one) depend on multi-parameter
mechanisms that are not \MIDR. By presenting an
implicit payment computation procedure that works
\emph{whenever} there exists a truthful mechanism
utilizing the given allocation function, we
believe that we have posed the multi-parameter
transformation at the appropriate level of
generality for future applications.

The literature contains surprisingly few examples
of truthful multi-parameter mechanisms that are
not based on \MIDR allocation rules.
Mechanisms designed by~\citet{BGN03,DNS06,DN11}
for various combinatorial auction domains make
use of what might be termed the \emph{pricing technique}:
each agent is allowed to choose freely
from a menu of alternatives, each specifying
an allocation and price. The menu presented to
a given agent may depend on the others' bids,
but it must be carefully
constructed so that self-interested agents
each choosing from their own menu will
never jointly select an infeasible allocation.
The taxation principle~\citep{Guesnerie81,Hammond79}
implies that \emph{every} dominant-strategy
truthful mechanism can actually be represented this way,
provided that agents are able to evaluate their own
utilities for different allocations before the
allocation is actually executed. In settings
with the no-simulation constraint,
the taxation principle does not apply because
agents can only evaluate their utility \emph{ex post}.
In the sponsored search setting, for example,
agents have no dominant strategy for choosing
from a menu listing bundles of ad impressions,
because without knowing CTRs they can't precisely
determine the value of an impression; on the
other hand, the mechanism is powerless to offer
a menu listing bundles of clicks, because
there is no way to guarantee that a bidder
who chooses a certain bundle will
receive the specified number of clicks.

Apart from mechanisms with \MIDR allocation
rules and those based on the pricing technique,
we are aware of only one other mechanism in the literature
that is dominant-strategy truthful in
a multi-parameter setting: the scheduling
mechanism of~\citet{LS07} for unrelated machines
that have only two possible processing times.
Their mechanism, like ours, is designed by
directly  constructing an allocation function
that satisfies the cyclic monotonicity constraints.

 \section{Preliminaries}
\label{sec:prelims}

\newcommand{\myProp}{\ensuremath{P}\xspace}

We study reductions from allocations to truthful mechanisms
for multi-parameter domains.
A CS-oriented background on multi-parameter mechanisms can be found in~\cite{ArcherKleinberg-sigecom08,ArcherKleinberg-ec08}, while an Economics-oriented background can be found in~\cite{Ashlagi-econometrica10}. Our main result holds for a very general framework for multi-parameter mechanisms, described below, where agents' types are defined as mappings from outcomes to valuations.
Our reduction invokes the allocation rule only once, which make it particularly useful in domains in which the allocation rule cannot be invoked (or simulated) more than once due to informational constraints.

\xhdr{Types, outcomes, and mechanisms.}
Multi-parameter mechanisms are defined as follows. There are $n$ agents and a set $\outcomes$ of outcomes. Each agent $i$ is characterized by his \emph{type}
    $\type_i: \outcomes \to \Re$,
where $\type_i(o)$ is interpreted as the agent's valuation for the outcome $o\in\outcomes$. For each agent $i$ there is a set of feasible types, denoted $\types_i$. Denote
    $\types = \types_1 \times \ldots \times \types_n$
and call it the \emph{type space};  call $\types_i$ the type space of agent $i$. The mechanism knows $(n,\outcomes,\types)$, but not the actual types $\type_i$;  each type $\type_i$ is known only to the corresponding agent $i$. Formally, a problem instance, also called a \emph{multi-parameter domain}, is a tuple $(n,\outcomes,\types)$.

A (direct revelation) mechanism $\mathcal{M}$ consists of the pair $(\mA, \mP)$,
where $\mA:\types\rightarrow \outcomes$ is the {\em allocation rule} and $\mP:\types\rightarrow \Re^n$ is the {\em payment rule}. Both $\mA$ and $\mP$ can be randomized, possibly with a common random seed. Each agent $i$ reports a type $\bid_i\in \types_i$ to the mechanism, which is called the {\em bid} of this agent. We denote the vector of bids by $\bid = (\bid_1 \LDOTS \bid_n) \in \types$. The mechanism receives the bid vector $\bid\in\types$, selects an outcome $\mA(\bid)$, and charges each agent $i$ a payment of $\mP_i(\bid)$.
The utilities are quasi-linear and agents are risk-neutral: if agent $i$ has type $\type_i \in \types_i$ and the bid vector is $\bid\in \types$, then this agent's  utility is
\begin{align}\label{eq:multiParam-util-defn}
u_i(\type_i;\bid) = \E_{\mM} \left[\, \type_i(\mA(\bid)) - \mP_i(\bid) \,\right].
\end{align}

For each type $\type_i\in \types_i$ of agent $i$ we use a standard notation $(\bid_{-i},\type_i)$ to denote the bid vector $\hat{\bid}$ such that $\hat{\bid}_i=\type_i$ and $\hat{\bid}_j = \bid_j$ for every agent $j\neq i$.

\xhdr{Game-theoretic properties.}
A mechanism is {\em truthful} if for every agent $i$ truthful bidding is a {\em dominant strategy}:
\begin{equation}\label{eq:multiParam-truthful}
u_i(\type_i; (\bid_{-i}, \type_i))
    \geq u_i(\type_i; \bid)
    \quad \forall \type_i\in \types_i, \;\bid \in \types.
\end{equation}
An allocation rule is called \emph{truthfully implementable} if it is the allocation rule in some truthful mechanism.

A mechanism is {\em individually rational (IR)} if each agent $i$  never receives negative utility by participating in the mechanism and bidding truthfully:
    \begin{equation}\label{eq:multiParam-IR}
        u_i(\type_i; (\bid_{-i}, \type_i)) \geq 0
            \quad \forall \type_i\in \types_i, \;\bid_{-i} \in \types_{-i}.
    \end{equation}

The right-hand side in~\eqref{eq:multiParam-IR} represents the maximal guaranteed utility of an ``outside option'' (i.e., from not participating in the mechanism). For example, our definition of IR is meaningful whenever this utility is $0$, which is a typical assumption for most multi-parameter domains studied in the literature.

Note that if the mechanism is randomized, the above properties are defined in expectation over the internal random seed. We can also define utility (and, accordingly, truthfulness and IR) for a given realization of the random seed. We say a mechanism is \emph{universally truthful} if it is truthful for all realizations of the random seed; similarly for IR and other properties.

\xhdr{Our assumptions.}
We make two assumptions on the type space $\types$:
\begin{itemize}
\item \emph{non-negative types}: $\type_i(o)\geq 0$ for each agent $i$, type $\type_i\in \types_i$, each outcome $o\in\outcomes$.

\item \emph{rescalable types}: $\lambda \type_i\in \types_i$
for each agent $i$, type $\type_i\in \types_i$, and any parameter $\lambda\in[0,1]$.
($\lambda \type_i$ denotes the type $\type'_i$ whose valuation for every outcome $o$ satisfies $\type'_i(o) = \lambda \type_i(o)$.)

\OMIT{\item The type space is \emph{convex}: for each agent $i$, any types $\type, \type'\in \types_i$, and any parameter $\lambda\in[0,1]$ it holds that
    $\lambda \type + (1-\lambda) \type' \in \types_i $.}
\end{itemize}

In particular, for each agent $i$ there exists a \emph{zero type}: a type $\type_i\in \types_i$ such that $\type_i(\cdot)\equiv 0$. Let us say that a mechanism is \emph{normalized} if for each agent $i$, the expected payment of this agent is $0$ whenever she submits the zero type.
For domains with non-negative types, it is desirable that all agents are charged a non-negative amount; this is known as the \emph{no-positive-transfers} property.

\xhdr{Dot-product valuations.}
An important special case is \emph{dot-product valuations}, where the type $\type\in \types_i$ of each agent $i$ can be decomposed as a dot product
    $\type(o) = \beta_{\type}\cdot a_i(o)$,
for each outcome $o\in\outcomes$, where
    $\beta_{\type}, a_i(o)\in \Re^{d}$
are some finite-dimensional vectors.
Here the term $a_i(o)$ is the same for all types $\type\in \types_i$ (and known to the mechanism), whereas $\beta_{\type}$ is the same for all outcomes $o\in \outcomes$ and is known only to agent $i$. The term $a_i(o)$ is usually called an ``allocation'' of agent $i$ for outcome $o$, and $\beta_{\type}$ is called the ``private value''. Single-parameter domains correspond to the case $d = 1$.

Note that the type $\type$ of each agent $i$ is determined by the corresponding private value $\beta_{\type}$, and his type space $\types_i$ is determined by
    $D_i = \{ \beta_{\type}:\; \type\in\types_i\} \subset \Re^d$.
Because of this, in the literature on dot-product valuations the term ``type'' often refers to $\beta_{\type}$. To avoid ambiguity, in this section we will refer to $\beta_{\type}$ as ``private value'' rather than ``type'', and call $D_1\times \ldots \times D_n$ the \emph{private value space}.

\OMIT{ 
In a domain with dot-product valuations, types are rescalable if and only if
$$\beta_{\type}\in D_i \Rightarrow \lambda \beta_{\type}\in D_i \;
\text{ for each }\lambda \in [0,1].
$$
In words, assuming rescalable types is equivalent to assuming that the set $D_i$ is star-convex at $0$.
} 

In a domain with dot-product valuations, types are rescalable if and only if for each $\beta_{\type}\in D_i$ and each $\lambda \in [0,1]$ it holds that
    $\lambda \beta_{\type}\in D_i$.
In other words, if and only if the set $D_i$ is star-convex at $0$. To ensure non-negative types, it suffices to assume that
    $D_i\subset \Re_+^{d}$ for each agent $i$, and all allocations are non-negative: $a_i(o) \in \Re_+^{d}$ for all $o \in \outcomes$.

\xhdr{Truthfulness characterization.}
We will use a characterization of truthful mechanisms via a property called ``cycle-monotonicity" (henceforth abbreviated as \CMON). A (randomized) allocation rule $\mA$ satisfies \CMON  if the following holds: for each bid vector $\bid\in\types$, each agent $i$, each $k\geq 2$, and each $k$-tuple
    $\type_{i,0},\; \type_{i,1} \LDOTS \type_{i,k} \in \types_i$
of this agent's types, we have
\begin{align}\label{eq:CMON}
\E_{\mA}\left[ \textstyle \sum_{j=0}^k \;
    \type_{i,j}\left( o_{i,j} \right)
    -     \type_{i,\,(j-1) \bmod{k}}\left( o_{i,j} \right)
    \right] \geq 0,\quad
    \text{ where }
    o_{i,j} = \mA\left(\bid_{-i},\, \type_{i,j} \right)\in \outcomes.
\end{align}
Recall that we are using a general notion of agents' types (and bids), which are defined as functions from outcomes to real-valued valuations.

It is known that $\mA$ is truthfully implementable if and only if it is cycle-monotone, in which case the corresponding payment rule is essentially fixed.

\begin{theorem}[\citet{Rochet-1987}]
\label{thm:CMON-characterization}
Consider an arbitrary multi-parameter domain $(n,\outcomes,\types)$. A (randomized) allocation rule $\mA$ is truthfully implementable if and only if it is cycle-monotone. Assuming rescalable types, for any cycle-monotone allocation rule $\mA$, a mechanism $(\mA,\mP)$ is truthful and normalized if and only if
\begin{align}\label{eq:multiParam-myerson}
\E_{\mA}\left[ \mP_i(\bid) \right] =
   \E_{\mA}\left[
    \bid_i(\mA(\bid)) -
     \textstyle \int_{t=0}^1 \;  \bid_i(\mA(\bid_{-i},\, t\, \bid_i )) \, dt
    \right].
\end{align}
\end{theorem}

This characterization generalizes a well-known truthfulness characterization of single-parameter mechanisms in terms of monotonicity, due to~\citep{Myerson,ArcherTardos}. Recall that for single-parameter domains, the type of each agent $i$ is captured by a single number (the private value $v_i$), and the outcome pertinent to this agent is also captured by a single number (this agent's allocation $a_i(o)$). The bid of agent $i$ is represented by $b_i\in \Re$. Cycle-monotonicity is then equivalent to a much simpler property called \emph{monotonicity}: for each agent, fixing the bids of other agents, increasing this agent's bid cannot decrease this agent's allocation. The payment formula~\refeq{eq:multiParam-myerson} can also be simplified, e.g. for non-negative valuations it is
\begin{align}\label{eq:myerson}
\mP_i(\bid)  =
 b_i\, \mA_i(b_{-i}, b_i) -
 \textstyle{\int^{b_i}_{0}\,  \mA_i(b_{-i}, u) \,du.}
\end{align}

\newcommand{\Dext}{\mD_{\mathtt{ext}}}

\xhdr{External seed.}
We allow allocation rules to receive input from the environment; a canonical example is pay-per-click auctions where such input consists of user clicks. Formally, the allocation rule and the payment rule depend on the additional argument $\omega$ which captures all relevant input from the environment. (To simplify the notation, we keep the dependence on $\omega$ implicit.) We call $\omega$ the \emph{external seed}, to distinguish from the internal random seed of the mechanism. We assume that $\omega$ is an independent sample from some fixed distribution $\Dext$; this distribution may be unknown to the mechanism.

All game-theoretic properties defined above carry over to mechanisms with external seed if all expectations are over both internal and external seed. In particular, Theorem~\ref{thm:CMON-characterization} carries over with no other modification.

We are primarily interested in properties that hold in expectation over the external seed, for all possible distributions $\Dext$ over the external seed. The corresponding version of a given property \myProp is called \emph{stochastically \myProp}. For example, we are interested in mechanisms that are stochastically truthful, and this requires the allocation rules to be stochastically \CMON.

We also define a stronger version of truthfulness: one that holds for each realization of the external seed. For each game-theoretic property \myProp described above, such as truthfulness, IR and \CMON, a version that holds for each realization of the external seed will be called \emph{ex-post \myProp}. Theorem~\ref{thm:CMON-characterization} holds for every given realization of the external seed (but requires the allocation rule to satisfy ex-post \CMON).

A crucial way in which the external seed is different from the internal randomness is that a given run of the allocation rule might not observe the entire external seed. More precisely, runs of the allocation rule on different bid vectors might observe different portions of the external seed. For example, if an ad is not displayed to a given user, the mechanism does not observe whether this user would have clicked on this ad if it were displayed. It follows that the mechanism might not be able to simulate the allocation rule on different bid vectors -- this is precisely the ``no-simulation" constraint discussed in the Introduction. Moreover, this issue can affect payment computation: the payment prescribed by~\eqref{eq:multiParam-myerson}, although well-defined as a mathematical expression, might not be computable given the information available to the mechanism.%
\footnote{This has been proved in~\citep{MechMAB-ec09,DevanurK09} in the context of multi-armed bandit mechanisms, see Section~\ref{sec:MAB-mech-problem} for more details.}

To address this issue formally, we say that the mechanism is \emph{information-feasible} if for each run of of the mechanism (i.e., for each bid vector $\bid$, each realization of the mechanism's internal randomness, and every possible value of the external seed) the payments are uniquely
determined given the information available to the mechanism.

\OMIT{ 
In the presence of the external seed, each game-theoretic property \myProp described above can have two versions: for each realization of the external seed (we call it \emph{ex-post} \myProp), and in expectation over the external seed (we call it \emph{stochastically \myProp)}. For example: ex-post truthful and stochastically truthful, ex-post \CMON and stochastically \CMON. In particular, Theorem~\ref{thm:CMON-characterization} holds for both versions, in a consistent way: in the ``ex-post'' version, all properties (truthful, cycle-monotone, and normalized) are ex-post, and in the ``stochastic'' version, all these properties are stochastic.}

\newcommand{\rM}{w} 
\newcommand{\rA}{s} 
\newcommand{\rP}{\bar{s}} 
\newcommand{\rN}{r} 
\newcommand{\rT}{q} 
\newcommand{\rY}{x} 
\newcommand{\rZ}{y} 

\newcommand{\CanonProc}{canonical self-resampling procedure}

\newcommand{\mba}{x}
\newcommand{\mbb}{y}

\xhdr{Implicit payment computation for single-parameter domains.}
\citet{Transform-ec10} provide an implicit payment computation result for single-parameter domains. They prove that any
monotone
allocation rule for any single-parameter domain can be transformed into a truthful, information-feasible mechanism with an arbitrarily small loss in expected welfare.
The allocation rule is only invoked once. Below we quote a special case of this result that is most relevant to the present paper.%
\footnote{We restate the result slightly, to make it consistent with our notation. \citep{Transform-ec10} states the mechanism more abstractly, in terms of a general \emph{self-resampling procedure}. The simple description of $\mM_\delta$ that we present here
was first published in~\citep{Parkes-netecon12}.}


\renewcommand{\algorithmcfname}{Mechanism}
\begin{algorithm}[t]
\caption{The single-parameter mechanism $\mM_\delta$ from~\citep{Transform-ec10}}
\begin{enumerate}
\item Collect bid vector $b$.
\item Independently for each agent $i \in [n]$, randomly sample
$\chi_i=1$ with probability $1-\delta$ and otherwise
$\chi_i=\gamma_i^{1/(1-\delta)}$, where $\gamma_i \in [0,1]$ is
sampled uniformly at random. \label{mechstep:rescale}
\item Construct the vector of modified bids, $ x = (\chi_1 b_1, \ldots, \chi_n b_n).$
\item Allocate according to $\tA(b) = \mA(x)$.
\item Compute payments using the formula
$
\tP_i(b) = b_i \cdot \mA_i(x) \cdot \begin{cases}
1 & \mbox{if $\chi_i = 1$} \\
1 - \frac{1}{\delta} & \mbox{if $\chi_i < 1$}
\end{cases}
$.
\end{enumerate}
\label{mech:BKS}
\end{algorithm}
\renewcommand{\algorithmcfname}{Algorithm}

\begin{theorem}[\citet{Transform-ec10}]\label{thm:single-parameter}
Consider an arbitrary single-parameter domain where the private values of each agent lie in the interval $[0,1]$. Let $\mA$ be a stochastically monotone allocation rule for this domain. Then for each $\delta\in(0,1)$, mechanism
    $\mM_\delta = (\tA,\tP)$
(described in Mechanism~\ref{mech:BKS}) is information-feasible and has the following properties.
\begin{itemize}
\OMIT{
\item[(a)] [Structure] Allocation rule $\tA$ is bid-resampling with respect to $\mA$, with resampling distribution that depends on $\delta$ but not on $\mA$.}

\item[(a)] [Incentives] $\mM_\delta$ is stochastically truthful, universally ex-post individually rational. If $\mA$ is ex-post monotone, then $\mM_\delta$ is ex-post truthful.
\item[(b)] [Performance] For $n$ agents and any bid vector $b$ (and any fixed external seed) allocations $\tA(b)$ and $\mA(b)$ are identical with probability at least $1-n\delta$. Moreover, if $\mA$ is $\alpha$-approximate (for social welfare),
then mechanism $\mM_\delta$ is $\alpha/(1-\frac{\delta}{2-\delta})$-approximate.

\item[(c)] [Payments] $\mM_\delta$ is ex-post \noPositiveTransfers; and although it is not universally so, for all realizations of the internal seed it never pays any agent $i$ more than
	$b_i\cdot \mA_i(\mba)\cdot (\frac{1}{\delta}-1)$.
    $\mM_\delta$ is universally ex-post normalized.
\end{itemize}
\end{theorem}

\section{The multi-parameter transformation}
\label{sec:transform}

In this section we present our first main contribution: the implicit payment computation result for multi-parameter domains. For a given multi-parameter domain and a given \CMON allocation rule for this domain,%
\footnote{Recall that \CMON is a necessary and sufficient
condition for truthfulness.}
our goal is to design a truthful, information-feasible mechanism
with outcome that is almost always identical to that of the original allocation rule, and this, in particular, ensures a small loss in expected welfare.
We achieve this goal for \emph{every} \CMON allocation rule and \emph{every} multi-parameter domain (under a mild assumption of rescalable, non-negative types). More precisely, we give a general ``multi-parameter transformation" which takes an arbitrary \CMON allocation rule $\mA$ and transforms it into a truthful, information-feasible mechanism which implements the same outcome as $\mA$ with probability arbitrarily close to $1$. This mechanism requires evaluating $\mA$ only once; its allocation rule randomly modifies the submitted bids, and then calls $\mA$ on the modified bids.%
\footnote{The transformation presented here is certainly not the only reduction that transforms multi-parameter allocation rules satisfying \CMON into truthful, information-feasible mechanisms. One appealing feature of our transformation, in comparison to alternatives, is its simplicity. It also optimizes the trade-off between the worst-case bid-to-payment ratio and the probability of adopting the original allocation, as was shown by Wilkens and Sivan (2012) in the single-parameter context.}
The technical contribution here is showing that the natural generalization of the reduction for the single-parameter setting, to the multi-parameter setting, preserves all desired properties.
The non-trivial part of the proof is showing that although the single-parameter transformation only ensures that each agent does not have an incentive to deviate by scaling all his bids by the same scalar in $[0,1]$, he also does not have an incentive to deviate to any other arbitrary bids.

\xhdr{The transformation.}
Our multi-parameter transformation is a remarkably straightforward
generalization of the single-parameter transformation specified in
Mechanism~\ref{mech:BKS}. In fact, there is no need to rewrite the
five steps; the only thing that changes is the interpretation of
the notation. Specifically, the bids $\bid_1,\ldots,\bid_n$ should
now be interpreted as elements of the type spaces
$\types_1,\ldots,\types_n$ rather than as scalars, and for each $i$
the modified bid $\type_i = \chi_i\, \bid_i$ is obtained by multiplying the abstract
type $\bid_i$ (a function from outcomes to reals) by the random scalar $\chi_i$.
(Note that $\chi_i\, \bid_i$ is well-defined because we are assuming the
rescalable types property.) The notation $b_i\cdot \mA_i(x)$ from the single-parameter case is now interpreted as
    $\bid_i (\mA_i(\type))$,
where
    $\type = (\type_1 \LDOTS \type_n)$
is the vector of re-sampled bids. With this interpretation, the payment rule is as follows:
\begin{align*}
\tP_i(\bid) = \bid_i (\mA_i(\type)) \cdot \begin{cases}
1 & \mbox{if $\chi_i = 1$} \\
1 - \frac{1}{\delta} & \mbox{if $\chi_i < 1$}
\end{cases}.
\end{align*}

In the remainder of this section we analyze the properties of
the multi-parameter transformation, proving an analogue of
Theorem~\ref{thm:single-parameter}. The subtlest step, which
occupies most of the analysis, is to prove that the modified
allocation rule $\tA$ satisfies \CMON.

\OMIT{
\xhdr{Generality.} We consider the general multi-parameter setting defined in Section~\ref{sec:prelims}, making two assumptions on the type space $\types$: non-negative, rescalable types.
}

\xhdr{Induced single-parameter domains.}
To aid in the analysis, it will be helpful to introduce the following
notation.
Consider 
a bid vector $\bid\in \types$ and a vector of ``rescaling coefficients''
$\lambda \in [0,1]^n$. Denote
    $$\lambda \otimes b = (\lambda_1 \bid_1 \LDOTS \lambda_n \bid_n ) \in \types. $$
In other words, $\lambda\otimes b$ is the rescaled bid vector where the bid of each agent $i$ is $\lambda_i \bid_i$.
Note that for each $\bid$ the subset
    $$\types_{\bid} = \{ \lambda\otimes b:\; \lambda \in [0,1]^n\} \subset \types$$
forms a single-parameter type space where each agent $i$ has private value $\lambda_i\in [0,1]$ and allocation $b_i(o)$ for every outcome $o$. By abuse of notation, let us treat the allocation and  payment rules for $\types_{\bid}$ as functions from the private value space $[0,1]^n$ rather than the type space $\types_{\bid}$.

\OMIT{
For intuition, consider an allocation rule
   $\mA_{\bid}(\lambda) = \mA(\lambda\otimes \bid)$
for the single-parameter type space $\types_{\bid}$. Note that if the original allocation rule $\mA$ is truthfully implementable for type space $\types$ using payment rule $\mP$, then $\mA_{\bid}$ is truthfully implementable for type space $\types_{\bid}$ using payment rule
   $\mP_{\bid}(\lambda) = \mP(\lambda\otimes \bid)$. This is because
restricting the allocation and payment rules to $\types_{\bid}$ only
limits the set of possible misreports of an agent. Therefore, we make an informal observation that solving our problem for $\mA$ essentially requires us to solve a similar problem for $\mA_{\bid}$. The latter can be done via the single-parameter transformation from~\cite{Transform-ec10}. In fact, we use this approach to solve the original problem for $\mA$: for each bid vector $\bid$ we apply Theorem~\ref{thm:single-parameter} to $\mA_{\bid}$.
}

We want to 
prove that the mechanism $\mM_\delta = (\tA,\tP)$ defined by our
transformation is truthful.
As a starting observation, note that when
one applies the single-parameter transformation given in
Section~\ref{sec:prelims} to the allocation rule defined
by    $\mA_{\bid}(\lambda) = \mA(\lambda\otimes \bid)$,
one obtains a mechanism that coincides with
the restriction of $\mM_\delta$ to $\types_{\bid}$.
By Theorem~\ref{thm:single-parameter}, we may
conclude that the restriction of $\mM_\delta$
to the single-parameter type space $\types_{\bid}$ is
truthful.
Yet this conclusion is not sufficient, since
this truthfulness condition is actually weaker than what we are aiming for: it ensures that a deviation inside the single-parameter type space $\types_{\bid}$ is not beneficial, but says nothing about deviation to other types in $\types\setminus \types_{\bid}$.
 Nevertheless, our proof will show that if the original allocation rule was \CMON, the transformed allocation rule is also \CMON for the domain $\types$, and thus is truthful as needed.

\OMIT{ 
\xhdr{The transformation.}
For each bid vector $\bid\in\types$ and a parameter $\delta\in (0,1)$, let
 $(\tA_{\bid},\tP_{\bid})$
be the single-parameter mechanism for type space $\types_{\bid}$ obtained by applying Theorem~\ref{thm:single-parameter} with parameter $\delta$ to allocation
    $\mA_{\bid}(\lambda) = \mA(\lambda\otimes \bid)$.%
\footnote{The transformed mechanism depends on $\delta$. To simplify presentation, we do not make this dependence explicit in the notation.}

Let $\mathbf{f}_{\mu} = (f_1 \LDOTS f_n)$ be the $n$-tuple of canonical self-resampling procedures with resampling probability $\mu$, for some fixed $\mu\in (0,1)$. (See Algorithm~\ref{alg:canon-proc} on page~\pageref{alg:canon-proc}.) For each bid vector $\bid\in\types$, let
\begin{align}\label{eq:multiParam-reduction-1}
 (\tA_{\bid},\tP_{\bid}) = \GenericMech(\mA_{\bid},\mu,\mathbf{f}_{\mu})
\end{align}
be the single-parameter mechanism for type space $\types_{\bid}$ obtained by applying the single-parameter transformation from~\cite{Transform-ec10} to allocation $\mA_{\bid}$.%
\footnote{Note that the transformed mechanism depends on $\mu$. We do not make this dependence explicit, to simplify the notation.}

The transformed multi-parameter mechanism $\mM_\delta$ is defined as
\begin{align}\label{eq:multiParam-reduction}
    \left( \tA(\bid),\; \tP(\bid) \right) = \left( \tA_{\bid}(\vec{1}\,),\; \tP_{\bid}(\vec{1}\,) \right)
        \text{ for every } b\in \types.
\end{align}
This completes the description of our multi-parameter transformation. The useful properties of this transformation are captured in the theorem below.
}

\begin{theorem}\label{thm:multiParam}
Consider an arbitrary multi-parameter domain $(n,\outcomes,\types)$ with rescalable, non-negative types. Let $\mA$ be a stochastically \CMON allocation rule for this domain. Let $\mM_\delta= (\tA,\tP)$ be the transformed mechanism 
for some parameter $\delta\in(0,1)$. Then $\mM_\delta$ has the following properties:
\begin{itemize}
\item[(a)] [Structure] $\mM_\delta$ is information-feasible.

\item[(b)] [Incentives] $\mM_\delta$ is stochastically truthful and universally ex-post individually rational.  If $\mA$ is ex-post \CMON, then $\mM$ is ex-post truthful.
\item[(c)] [Performance] For $n$ agents and any bid vector $b$ (and any fixed external seed) allocations $\tA(b)$ and $\mA(b)$ are identical with probability at least $1-n\delta$. Moreover, if $\mA$ is $\alpha$-approximation to the maximal social welfare then $\tA$ is $\alpha/\left( 1-\tfrac{2}{1-\delta} \right)$-approximation to the maximal  social welfare.

\item[(d)] [Payments] $\mM$ is ex-post \noPositiveTransfers; and although it is not universally so, for all realizations of the internal seed it never pays any agent $i$ more than
    	$\bid_i(o) (\frac{1}{\delta}-1)$, where $o=\mA(\bid)\in\outcomes$.
    Additionally, $\mM$ is universally ex-post normalized.

\end{itemize}
\end{theorem}

\begin{proof}
$\mM_\delta$ is information-feasible by construction, since so are the single-parameter mechanisms obtained from Theorem~\ref{thm:single-parameter}. All claimed properties except truthfulness follow immediately from Theorem~\ref{thm:single-parameter}. Below we prove truthfulness.

\OMIT{ 
By Theorem~\ref{thm:single-parameter}(a), the single-parameter allocation rule $\tA_{\bid}$ has the following property: for each agent $i$ the single-parameter bid $\lambda_i$ is rescaled by the (randomly chosen) factor $\chi_i\in[0,1]$ which does not depend on the bid, and then $\mA_{\bid}$ is called. Therefore, letting $\chi = (\chi_1\LDOTS \chi_n)$, it holds that
\begin{align}\label{eq:lm:multiParam-reduction-assumption}
    \tA_{\bid}(\lambda) = \mA( \chi\otimes (\lambda \otimes \bid) )\quad
        \text{ for all $\bid\in \types$ and $\lambda\in [0,1]^n$}.
\end{align}
} 

We claim that $\tA$ satisfies \CMON. Indeed, fix bid vector $\bid\in\types$, agent $i$, some $k\geq 2$, and a $k$-tuple
 $\type_{i,0},\; \type_{i,1} \LDOTS \type_{i,k} \in \types_i$
of this agent's types. Let us consider a fixed realization of the random vector $\chi \in [0,1]^n$ in step (\ref{mechstep:rescale}) of mechanism $\mM_\delta$.
For each type $\type_{i,j}$, note that
we have
$$ \tA(\type_{i,j}, \bid_{-i})
    = \mA\left( \chi \otimes (\type_{i,j},\,\bid_{-i}) \right) \in \outcomes.
$$
Denote this outcome by $o_{i,j}(\chi)$. Let us apply the cycle-monotonicity of $\mA$ for bid vector $\chi \otimes (\type_{i,j}, \bid_{-i})$:
\begin{align}\label{eq:proof-CMON}
\E_{\mA}\left[ \textstyle \sum_{j=0}^k \; \type_{i,j}(o_{i,j}(\chi)) -
    \type_{i,\;(j-1) \bmod{k}}(o_{i,j}(\chi)) \right]\geq 0.
\end{align}
Recalling that
    $o_{i,j}(\chi) = \tA(\type_{i,j}, \bid_{-i})$,
we observe that for this fixed realization of $\chi$,~\eqref{eq:proof-CMON} is exactly the inequality in the definition of cycle-monotonicity for $\tA$. Therefore taking expectation over $\chi$, we obtain the desired inequality~\eqref{eq:CMON} for $\tA$. Claim proved.%
\footnote{Note that the proof of cycle-monotonicity of $\tA$ did not use any other property of $M_\delta$ other than that the re-scaling factors $\chi_i$ are chosen independently from a distribution that does not depend on $\mA$. The truthfulness of the single-parameter mechanisms $(\tA_{\bid},\tP_{\bid})$  is used in the forthcoming argument about payments.}

\newcommand{\hA}{\widehat{\mA}}
\newcommand{\hP}{\widehat{\mP}}

It remains to prove that in the transformed mechanism $(\tA,\tP)$, the payment rule satisfies~\eqref{eq:multiParam-myerson}. Fix bid vector $\bid$ and consider the transformed single-parameter mechanism $(\tA_{\bid},\tP_{\bid})$ for the single-parameter type space $\types_{\bid}$. In the terminology of single-parameter domains, each agent $i$ receives an allocation
    $\tA_{\bid,\,i}(\lambda) = b_i(\tA_{\bid}(\lambda))$
whenever the bid vector is
    $\lambda \in [0,1]^n$.
Since this is a truthful and normalized single-parameter mechanism, it follows that
\begin{align*}
\E\left[ \tP_{\bid}(\lambda) \right]  =
\E\left[ \lambda_i\, \tA_{\bid,\,i}(\lambda) -
 \int_0^{\lambda_i}\,  \tA_{\bid,\,i}(\lambda_{-i}, t) \,dt\right],
 \quad \forall \lambda\in [0,1]^n.
\end{align*}
Plugging in $\lambda = \vec{1}$
and using the definitions of $\tA_{\bid}, \tP_{\bid}$,
we obtain the desired~\eqref{eq:multiParam-myerson}.
\end{proof}

\section{Multi-parameter MAB mechanisms}
\label{sec:MAB-mech-problem}

Let us define a natural multi-parameter extension to the MAB mechanism design problem studied in~\citep{MechMAB-ec09,DevanurK09,Transform-ec10}.%
\footnote{Here and elsewhere, \emph{MAB} stands for \emph{multi-armed bandits}.}

\xhdr{Problem formulation.}
There are $n$ agents. For each agent there is a known and fixed set of ads he is interested in; we assume that these sets are disjoint.
The total number of ads is denoted by $m$.

As is  common in the literature on sponsored search we assume that agents only value clicks; they have no value for an impression when the ad is not clicked.
For every ad $j$ there is a \emph{value-per-click} $v_j$ such that the unique agent  that is interested in that ad receives utility $v_j$ whenever this ad is clicked; this value is the agent's private information.

A mechanism for this domain proceeds as follows. There are $T$ rounds, where the time horizon $T$ is fixed and known to everyone. In each round the mechanism either decides to \emph{skip} this round or chooses one ad to display. Then the ad is either clicked or not clicked. All agents bid once, before the first round.
 The bid of a given agent consists of a tuple of reported values for his ads. The bid reported for ad $j$ is denoted $b_j$; the entire bid vector of all agents for the $m$ ads is denoted $b = (b_1 \LDOTS b_m)$.
Payments are assigned after the last round.

For each ad $j$, the click probability is fixed over time and denoted $\mu_j$.
In each round when this ad is displayed, it is clicked independently with probability $\mu_j$.
Click probabilities are called \emph{click-though rates} (\emph{CTRs}) in the industry. We assume that the CTRs are not known neither to the mechanism nor to the agents. For brevity, let
    $\mu = (\mu_1 \LDOTS \mu_m)$
be the vector of all CTRs.

\xhdr{Interpretation as a multi-parameter domain.}
For our setting, stochastic truthfulness (and similarly stochastic \CMON, etc.) is a property that holds in expectation over clicks, for all possible CTR vectors $\mu$.

Following the prior work, the external seed is defined as \emph{click realization} $\rho$, in the following sense. For every ad $j$ and every round $t$, realization $\rho(t,j)\in \{0,1\}$ says whether this ad would be clicked if it is shown in this round. In particular, ex-post truthfulness corresponds to truthfulness for every click realization. Note that a given run of a mechanism does not observe the entire click realization:
it only observes clicks for ads that are displayed in a given round.

For every bid vector $b$ and each click realization $\rho$, let $C_j(b,\rho)$ be the expected total number of clicks received by ad $j$, where the expectation is over the internal randomness in the mechanism. Denote
    $C(b,\rho) = (C_1(b,\rho) \LDOTS C_m(b,\rho))$
and call it the \emph{click vector}. We interpret the click vectors as the ``outcomes'' in the multi-parameter domain. Note that a given click vector $C(b,\rho)$ corresponds to expected welfare
    $\sum_j v_j C_j(b,\rho)$.

Note that with this interpretation of the ``outcomes'', the allocation rule is not free to choose any well-defined outcome. Instead, the collection of outcomes that can be implemented on a given run of the mechanism is constrained by the click realization.%
\footnote{Alternatively, we could have defined ``outcomes'' via impressions rather than clicks.
But then an agent would not have a full knowledge of his value for each outcome (his type) as the CTRs are not known to him.
Such a definition necessitates some cumbersome modifications to the framework in Section~\ref{sec:prelims}. Both versions lead to the same results.}

For a given CTR vector $\mu$, let
    $C(b,\mu) = \E_{\rho\sim \mu} C(b,\rho)$,
where the expectation is taken over click realizations $\rho$ according to the corresponding CTRs. Call it the $\mu$-click vector. In expectation over the clicks, the welfare is $\sum_j v_j C_j(b,\mu)$. When considering stochastic truthfulness, it will be more convenient to re-define outcomes as $\mu$-click vectors.

\xhdr{Discussion and background.}
If not for the issue of incentives and the requirement of truthfulness, the welfare-maximization problem for the allocation rule is precisely the  \emph{multi-armed bandit} problem (henceforth, \emph{MAB}), a well-studied problem in Machine Learning and Operations Research.
MAB \emph{mechanisms} can be seen as a version of the MAB problem that incorporates incentives. MAB mechanisms (in the limited single-parameter case, with one ad per agent), were introduced and studied in~\citep{MechMAB-ec09,DevanurK09} for the deterministic case. Subsequently, \citet{Transform-ec10} studied randomized MAB mechanisms.
Below we recap some of the contributions made in~\citep{MechMAB-ec09,DevanurK09}.

MAB mechanisms were suggested as a simple model in which one can study the interplay between incentives and learning, two major issues that arise in pay-per-click auctions.
Pay-per-click
is (along with pay-per-impression) one of the two prevalent business models in the advertising on the Internet, and \emph{the} prevalent pricing model in sponsored search.
Compared to pay-per-impression, pay-per-click reduces the risk that advertisers take, as they only pay when the ad is clicked. The seller, who has some control over clicks, bears the risk instead.  Moreover, advertisers typically have very little or no information about their CTRs, and should not be required to learn more. The pay-per-click model essentially shields the advertisers from this uncertainty.

The crucial assumption in
our model of
MAB mechanisms is that the CTRs are initially not known to the mechanism. This assumption reflects the fact that the CTRs are learned over time, while the ads are being allocated, and so the process of learning should be treated as a part of the game.%
\footnote{If some information on CTRs is known before the allocation starts, this can be modeled via Bayesian priors on CTRs. Following~\citep{MechMAB-ec09,DevanurK09,Transform-ec10}, we focus on the non-Bayesian version.}

The focus of the investigation in~\citep{MechMAB-ec09,DevanurK09} was whether and how the requirement of truthfulness restricts the performance of MAB algorithms
when types are single-parameter.
They found a very severe restriction for deterministic, ex-post truthful mechanisms: the allocation rule can only have a very simple, ``na\"ive'' structure (separating exploration and exploitation),
which severely impacts performance compared to the best MAB algorithms. They capitalize on the ``no-simulation'' constraint to prove that if an allocation rule does not conform to this simple structure, then a truthful mechanism with this allocation rule cannot be information-feasible.

The obstacle of information-feasibility for the single parameter case is circumvented in~\citet{Transform-ec10}
by moving from deterministic to randomized MAB mechanisms.
The single-parameter transformation (Theorem~\ref{thm:single-parameter}) reduces the design of truthful, information-feasible MAB mechanisms to the design of monotone allocation rules for this domain. Further, the authors provide monotone allocation rules whose performance matches that of optimal MAB algorithms. Specifically, they show that (a minor modification of) a standard MAB algorithm \UCB~\citep{bandits-ucb1} is stochastically monotone, and they design a new MAB algorithm which is ex-post monotone and has essentially the same performance.
\OMIT{

\xhdr{Initial ideas.}
The problem of designing truthful mechanisms for welfare maximization in multi-parameter domains is known to be very challenging whenever the VCG mechanism is not feasible (as in combinatorial auctions, when VCG is computationally
infeasible).
For our problem, if the CTRs were public information then the VCG mechanism could have been used to create a truthful mechanism that maximizes the social welfare. Our problem is challenging as the CTRs are not known and this creates an obstacle of information-feasibility, which in particular precludes the use of the VCG mechanism. We next discuss some simple approaches to create truthful mechanisms: the first disregards the bids, and the second uses randomization to reduce the problem to a single parameter problem.

The first obvious approach is \emph{bid-independent} allocation rules -- ones that do not depend on the bids. Among those,
we naturally focus on the allocation rule that samples an ad independently and uniformly at random in each round; call it $\Rand$. It is easy to see that $\Rand$ achieves the best worst-case performance among all bid-independent allocation rules. Namely, if a bid-independent allocation rule chooses some ad $j$ in less than $\tfrac{T}{m}$ rounds, then it may be the case that $b_j$ is high, whereas the bids on all other ads are very low.

A slightly more complicated approach randomly reduces the problem to a single parameter problem as follows.
One ad is selected independently for each agent, uniformly at random from this agent's ads.
Then some truthful single-parameter mechanism $\mM$ is run on the selected ads. Call this mechanism $\MosheMech$.
\footnote{{\bf (Alex: \emph{please} suggest a better name than $\MosheMech$. Use the macro $\backslash$MosheMech, defined in the main file.)}}
This mechanism is truthful (ex-post or stochastically, same as $\mM$) because
for each realization of the selection described above, it is simply a truthful single-parameter mechanism.
The performance of this mechanism is the same as the performance of the trivial $\Rand$ mechanism when there is only one agent.
}

\OMIT{%
\xhdr{Our contribution.}
Given the poor performance of the two simple mechanism above, we would like to construct real multi-parameter mechanism that do not simply reduce the problem to a single parameter problem, and achieve better performance for interesting inputs, for example, for a single agent problem.
The most success approach in the literature for the construction of such mechanisms is using Maximal-In-Distributional Range (MIDR) mechanisms.
Unfortunately, such mechanisms are not feasible for our problem due to the obstacle of information-feasibility (MOSHE: DO WE HAVE A PROOF FOR THAT?). Thus we are left with the challenge of designing CMON allocation rules which outperform the simple rules, at least for large set of interesting inputs.  Coming up with CMON allocation rules is known to be very challenging and there are only few successes we are aware of (CITE).
In Section~\ref{sec:positive} we present a stochastically CMON allocation rule with improved performance. Combined with our implicit payment computation reduction for multi-parameter domains we obtain a stochastically truthful mechanism.
Before that, in Section~\ref{sec:negative} we show that for a single agent, if we insist on ex-post trustfulness then $\Rand$ is essentially optimal.

In Appendix we discuss several 
approaches that seem appealing but do not result in \CMON allocation rules.
}

\OMIT{


\xhdr{Simple approaches.}
As with any new problem, it is
useful
to consider simple solutions before attempting to design smarter ones. In this case, we are interested in allocation rules that are obviously \CMON.

One obvious approach is \emph{bid-independent} allocation rules -- ones that do not depend on the bids. Among those,
we naturally focus
on the allocation rule that samples an ad independently and uniformly at random in each round; call it $\Rand$. It is easy to see than $\Rand$ achieves the best worst-case performance among all bid-independent allocation rules. Namely, if a bid-independent allocation rule chooses some ad $j$ in less than $\tfrac{T}{m}$ rounds, then it may be the case that $b_j$ is high, whereas the bids on all other ads are very low.

A slightly more complicated approach
essentially reduces the problem to a single parameter problem as follows.
The rounds are partitioned in contiguous phases. In each phase, one ad is selected independently for each agent, uniformly at random from this agent's ads. Then some truthful single-parameter mechanism $\mM$ is run on the selected ads for the duration of the phase. Call this mechanism $\MosheMech$.%
\footnote{{\bf (Alex: Moshe, please pick a better name than $\MosheMech$. Use the macro $\backslash$MosheMech, defined in the main file.)}}
This mechanism is truthful (ex-post or stochastically, same as $\mM$) because in each phase, for each realization of the selection described above, it is simply a truthful single-parameter mechanism.
The performance of this mechanism is the same as the performance of the trivial $\Rand$ mechanism when there is only one agent.

In Appendix we discuss several trivial approaches that seem appealing but do not result in \CMON allocation rules.
}


\section{Multi-parameter MAB mechanisms: \\Impossibility result for ex-post truthful mechanisms}
\label{sec:negative}

In this section we present our second main contribution: a very strong impossibility result for ex-post truthful multi-parameter MAB mechanisms. Consider one of the agents and fix the bids of the others. Essentially, we show that an  allocation rule which satisfies ex-post \CMON for that agent, cannot depend on the bid of that agent. More precisely, this holds for a deterministic allocation rule if the bids are large enough, as well as for any allocation rule (deterministic or randomized) that never skips a round. For randomized allocation rules that may skip a round, we show that if the allocation rule satisfies ex-post \CMON then it cannot achieve a nontrivial worst-case approximation ratio.

\OMIT{An allocation rule is called \emph{scale-free} if scaling the bid vector $b$ by a scalar $\lambda >0$ does not change the allocation (assuming that $\lambda b$ is a valid vector of bids). Intuitively, this is a mild and reasonable assumption: e.g., we do not expect the outcome to depend on whether the bids are priced in dollars or in euro.}

\begin{theorem}\label{thm:ExPostWMON}
Let $\mA$ be an allocation rule for multi-parameter MAB which satisfies ex-post \CMON. Fix any agent $i$, and fix bids submitted by all other agents.
\begin{itemize}
\item[(a)] If $\mA$ is any allocation rule (deterministic or randomized) that
never skips a round, and if agent $i$ is the only agent,
then the allocation has no dependence on his bids.
\item[(b)] If $\mA$ is deterministic,
then there exists a finite $B$ such that the allocation for agent
$i$ does not depend on his bids, as long as all his bids are larger than $B$.
\item[(c)] If $\mA$ is randomized, then its worst-case approximation
ratio (over all bid vectors of agent $i$) is no better than that of the
trivial randomized allocation rule that ignores agent $i$'s bid, samples
one of his ads uniformly at random, and allocates all impressions to
that ad.
 \end{itemize}
\end{theorem}

The first conclusion presumes there is only a single agent, and to prove the remaining two conclusions it suffices to consider the case of a single agent, because from the perspective of any given agent the ads allocated to other agents can be represented as skips. (In particular, allowing skips in single-agent allocation rules is essential for the generalization to multiple agents.) In the rest of this section we assume a single agent with $m$ ads.

To prove our result we need to set up some notation. Recall that the bids of the agent are represented by a vector
    $b=(b_1 \LDOTS b_m) \in\R^m_+$.
For a given allocation rule $\mA$ and a given click-realization $\rho$, the \emph{impression allocation} $\mA(b,t,\rho) \in \R^m_+$ is a vector of probabilities, in expectation over the random seed of the algorithm, so that
$\mA_i(b,t,\rho)$ is the probability that ad $i$ is chosen in round $t$ given bid vector $b$ and realization $\rho$.

\xhdr{Weak monotonicity.}
We use \CMON through a special case where $k=2$ in~\eqref{eq:CMON}; this special case is known in the literature as \emph{weak monotonicity}, henceforth abbreviated \WMON. \WMON is equivalent to \CMON if there are finitely many outcomes and the type space is convex~\citep{WMON-SaksYu}.  It follows that in our setting, ex-post \WMON is equivalent to ex-post \CMON for deterministic allocation rules. For more background on \WMON, see~\citep{ArcherKleinberg-ec08}.

Let us restate \WMON in the notation of multi-parameter MAB mechanisms. Recall that the click vector $C(b,\rho)$ is a vector such that $C_j(b,\rho)$ is the total expected number of clicks for ad $j$, given bid vector $b$ and realization $\rho$. Then
$$ C_j(b,\rho) \textstyle
    = \sum_{t=1}^T \; \rho(t,j)\, \mA_j(b,t,\rho)
    = \sum_{t=1}^T \; \Delta_t(\rho)\, \mA(b,t,\rho),
$$
where $\Delta_t(\rho)$ is the $m\times m$ diagonal matrix with diagonal entries
    $(\rho(t,1) \LDOTS \rho(t,m))$.
Ex-post \WMON states the following: for any realization $\rho$ and any bid vectors $b,\bt\in \R^m_+$,
$$ (\bt-b)\cdot (\;
    C(\bt,\rho) - C(\bt,\rho)
    \;) \geq 0
$$
Re-writing this in terms of the impression allocation, we obtain:
\begin{align}\label{eq:WMON-alloc}
(\bt-b)^\dag\;
    \textstyle\sum_{t=1}^T \; \Delta_t(\rho) (\;
        \mA(\bt,t,\rho) - \mA(b,t,\rho)
    \;)\geq 0.
\end{align}
Here and elsewhere, $M^\dag$ denotes a transpose of a matrix $M$.

\xhdr{Analysis for allocation rules with no skips (Theorem~\ref{thm:ExPostWMON}(a)).}
For the sake of contradiction, assume that
    $\mA(b,t,\rho) \neq \mA(b',t,\rho)$
for some round $t$, click-realization $\rho$, and bid vectors
    $b, b' \in \R^m_+ $.
Pick the smallest $t$ for which such counterexample exists. Assume w.l.o.g.\ that $\rho \equiv 0$ for all rounds
after $t$.
For each ad $i$, let $\rho_i$ be a realization that coincides with $\rho$ on all rounds but $t$, and in round $t$ ad $i$ is clicked and all other ads are not clicked.

Let $\bt = \vec{1}+\max(b,b') \in \R^m_+$, where $\max(b,b')$ is the coordinate-wise maximum of $b$ and $b'$.
Since
    $\mA(b,t,\rho) \neq \mA(b',t,\rho)$,
we can w.l.o.g. assume that
           $\mA(\bt,t,\rho) \neq \mA(b,t,\rho)$.
Since $\mA$ never skips a round,
\begin{equation} \label{eq:5.1.a.2}
\textstyle           \sum_{i=1}^m \mA_i(\bt,t,\rho) = 1 = \sum_{i=1}^m \mA_i(b,t,\rho).
\end{equation}
Combining
    $\mA(\bt,t,\rho) \neq \mA(b,t,\rho)$
with~\eqref{eq:5.1.a.2} we deduce that
for some ad $i$,
    $\mA_i(\bt,t,\rho) < \mA_i(b,t,\rho).$
We claim that \WMON is violated for bids $b,\bt$ and realization $\rho_i$. Indeed, consider~\eqref{eq:WMON-alloc} for realization $\rho_i$.
The sum in~\eqref{eq:WMON-alloc} is 0 for all rounds other than $t$ because $\mA(\bt,s,\rho) = \mA(b,s,\rho)$
for all rounds $s< t$ (by minimality of $t$), and $\rho_i\equiv 0$ for all rounds $s>t$. For round $t$, the sum in~\eqref{eq:WMON-alloc} is $0$ for all ads other than $i$, by definition of $\rho_i$. Thus, the sum is simply equal to
    $(b_i-\bt_i) \cdot [\mA_i(b,t,\rho) - \mA_i(\bt,t,\rho)]$, which is negative, contradicting~\eqref{eq:WMON-alloc}.

\xhdr{Analysis for the deterministic case (Theorem~\ref{thm:ExPostWMON}(b)).} We now address deterministic allocation rules that may skip rounds. The analysis of this case captures the main ideas of the randomized case while being significantly easier to present.

Fix click-realization $\rho$ and round $t$. Let $\mA$ be the deterministic allocation rule for agent $i$ that is induced by fixing the bids of all other agents. If $\mA$ skips round $t$, write $\mA(b,t,\rho) = \Skip$. For a vector $b = (b_1 \LDOTS b_m) \in \R^m_+$,
denote
    $\max(b\,) = \max_{1\leq i\leq m} b_i$.
Define $\min(b)$ similarly.

\OMIT{
We will use the following shorthand.
For a vector $b = (b_1 \LDOTS b_m) \in \R^m_+$,
denote
    $\min(b\,) = \min_{1\leq i\leq m} b_i$.
For two vectors $b,\bt\in \R^m_+$, let
    $\min(b,\bt) \in \R^m_+$
be the coordinate-wise minimum of $b$ and $\bt$. Use the same notation for $\max(b)$ and $\max(b,\bt)$. Let $b \succeq \bt$ denote coordinate-wise ``$\geq$", that is,
    $b_i \geq \bt_i$ for all $i$.

Fix click-realization $\rho$ and round $t$. Let $\mA$ be a deterministic allocation rule for a single agent. If $\mA$ skips round $t$, write $\mA(b,t,\rho) = \Skip$.
}

One technicality in the analysis is handling skips; we deal with it using the following notions:%
\footnote{We use a standard convention that $\sup(\emptyset)=-\infty$.}
\begin{align}
\bMin(t,\rho) &=  \sup\{ \max(b) :\;
    \text{$b \in \R^m_+$ and $\mA(b,t,\rho) = \Skip$} \}. \nonumber \\
B &= \max\left(\quad \{0\} \cup \left\{
    \bMin(t,\rho):\;
    \text{$\exists\, t,\rho$ such that $\bMin(t,\rho)<\infty $}
    \right\} \quad \right). \label{eq:negative-det-B}
\end{align}
Note that $B=0$ if $\bMin(t,\rho)=\infty$ for all $t$ and $\rho$.
For a given round $t$ and realization $\rho$, $\bMin(t,\rho)$ is defined such that if all $m$ bids are larger than $\bMin(t,\rho)$ then the allocation does not skip at round $t$ on realization $\rho$. $B$ is defined such that for every realization and every round, if all bids are larger than $B$ then the allocation rule never skips.

\begin{claim}\label{cl:ExPostWMON-det-skip}
Let $\mA$ be a deterministic single-agent allocation rule which satisfies ex-post \WMON.
Then for each click-realization $\rho$ and each round $t$, $\mA$ does not depend on the bid vector $b$ for all bid vectors
$b \in (B, \infty)^m$, where $B$ is defined in~\eqref{eq:negative-det-B}.
\end{claim}

\begin{proof}
For the sake of contradiction, assume that
    $\mA(b,t,\rho) \neq \mA(b',t,\rho)$
for some round $t$, click-realization $\rho$, and bid vectors
    $b, b' \in (B, \infty)^m $.
Pick the smallest $t$ for which such counterexample exists. Assume w.l.o.g. that $\rho \equiv 0$ for all rounds
after $t$.
For each ad $i$, let $\rho_i$ be a realization such that it coincides with $\rho$ on all rounds but $t$, and in round $t$ ad $i$ is clicked and all other ads are not clicked.

Let us consider two cases, depending on whether $\bMin(t,\rho)$ is finite.

\xhdr{Case 1: $\bMin(t,\rho) = \infty$.} At least one of
    $\mA(b,t,\rho)$, $\mA(b',t,\rho)$
is not equal to \Skip.
Since
    $\mA(b,t,\rho) \neq \mA(b',t,\rho)$,
we can w.l.o.g. assume that $\mA(b,t,\rho) \neq \Skip$.
Hence, $\mA_i(b,t,\rho)=1$ for some ad $i$.
Since $\bMin(t,\rho) = \infty$, there exists
    $\bt \in ( \max(b\,),\; \infty)^m$
such that $\mA(\bt,t,\rho) = \Skip$.

We claim \WMON is violated for bids $b,\bt$ and realization $\rho_i$. As in the first case, we see that the sum in~\eqref{eq:WMON-alloc} is 0 for all rounds other than $t$, and for round $t$ the sum is $0$ for all ads other than $i$. Again, it follows that the sum is simply equal to $b_i-\bt_i$, which is negative, contradicting~\eqref{eq:WMON-alloc}. Claim proved.

\xhdr{Case 2: $\bMin(t,\rho) < \infty$.}
The proof of this case is very similar to the proof of Theorem~\ref{thm:ExPostWMON}(b).

Recall that in case 1 it holds that $\bMin(t,\rho) < \infty$.
Let $\bt = \vec{1}+\max(b,b') \in \R^m_+$, where $\max(b,b')$ is the coordinate-wise maximum of $b$ and $b'$.
Since
    $\bMin(t,\rho) < \infty$,
it follows that
    $B\geq \bMin(t,\rho)$,
so neither $\mA(b,t,\rho)$ nor $\mA(b',t,\rho)$ nor $\mA(\bt,t,\rho)$ is equal to \Skip.
Since
    $\mA(b,t,\rho) \neq \mA(b',t,\rho)$,
we can w.l.o.g. assume that
    $\mA(\bt,t,\rho) \neq \mA(b,t,\rho)$.
In particular,
    $\mA_i(\bt,t,\rho)=0$
and
    $\mA_i(b,t,\rho) = 1$
for some ad $i$.

We claim that \WMON is violated for bids $b,\bt$ and realization $\rho_i$. Indeed, consider~\eqref{eq:WMON-alloc} for realization $\rho_i$. The sum in~\eqref{eq:WMON-alloc} is 0 for all rounds other than $t$ because $\mA(\bt,s,\rho) = \mA(b,s,\rho)$
for all rounds $s< t$ (by minimality of $t$), and $\rho_i\equiv 0$ for all rounds $s>t$. For round $t$, the sum in~\eqref{eq:WMON-alloc} is $0$ for all ads other than $i$, by definition of $\rho_i$. Thus, the sum is simply equal to $b_i-\bt_i$, which is negative, contradicting~\eqref{eq:WMON-alloc}. Claim proved.
\end{proof}

\subsection{Analysis of the randomized case: proof of Theorem~\ref{thm:ExPostWMON}(c)}
\label{thm:ExPostWMON-c-proof}

The proof of the randomized case of Theorem~\ref{thm:ExPostWMON} is technically  more involved than the proof of Theorem~\ref{thm:ExPostWMON}(b)).
In particular, even stating the analog of Claim~\ref{cl:ExPostWMON-det-skip} requires a considerable amount of setup.

\newcommand{\dmn}{{\mathcal{D}}}

Define functions
$f,g,G: \N \times \R_+ \to \R_+$ by the following recurrence:
    $f(0,y)=g(0,y)=G(0,y)=0$ for all $y$; while for $t>0$:
\begin{align*}
f(t,y) &= 3ym\,G(t-1,y) + 1\\
g(t,y) &= 2\,f(t,y) + 2 + 3ym\,G(t-1,y)\\
G(t,y) &= g(t,y) + G(t-1,y) = \textstyle{\sum_{s=1}^t g(s,y)}.
\end{align*}
For real numbers $B \geq 0,y \geq 1$, let $\dmn(B,y)$ denote the set
\[
\dmn(B,y) = \{b \, : \, \min(b\,) \geq B, \max(b)/\min(b) \leq y\}.
\]
We will refer to a bid vector as ``$y$-balanced'' if it satisfies
$\max(b)/\min(b) \leq y$.

Let $\mA$ be a (potentially randomized) allocation rule. Fix realization $\rho$. For all times $t$ and all $\eps>0, y \geq 1$ let
\begin{align*}
\AMax(t,\rho,y) &= \limsup_{x\to \infty}
    \left\{\; \Norm{ \mA(b, t,\rho) }\;:\; b\in \dmn(x,y) \; \right\} \\[6pt]
\bMax(t,\rho,\eps,y) &= \sup\left\{
    x :\; \exists b\in \dmn(x,y) \quad
    \Norm{ \mA(b,t,\rho) } < \AMax(t,\rho,y)- \eps\, f(t,y)\;
    \right\} \\[6pt]
B_{\eps}(y) &= \begin{cases}
0 & \mbox{if $\bMax(t,\rho,\eps,y) = \infty$ for all $t,\rho,\eps$} \\
\sup \left( \R \cap \{\bMax(t,\rho,\eps,y):\; t \in \N, \eps > 0\} \right) & \mbox{otherwise.}
\end{cases}
\end{align*}
Here $\AMax(t,\rho,y)$ is the maximal expected number of impressions at time $t$ such that the agent can obtain this number with arbitrarily large $y$-balanced bids. The meaning of $\bMax$ is as follows: if every component of a $y$-balanced vector $b$ is above $\bMax$, the expected number of impressions for time $t$ is guaranteed to be within $\eps\, f(t,y)$ of the best possible. Note that $B_\eps(y)=0$ if $\bMax(t,\rho,\eps,y)$ is infinite for all $t$ and all $\rho$.

\begin{claim}\label{cl:ExPostWMON-rand}
Let $\mA$ be a single-agent allocation rule which satisfies ex-post \WMON.
Then for any $y \geq 1$, any $\eps>0$,  any bid vectors
    $b,b' \in \dmn(B_{\eps}(y),y)$,
any realization $\rho$, and any round $t$,
we have
    $$\Norm{ \mA(b,t,\rho) - \mA(b',t,\rho) } \leq \eps\,g(t,y)$$
\end{claim}

\begin{proof}
Fix $\eps>0$ and realization $\rho$.  Let us use induction on $t$. Case $t=0$ is trivial, interpreting
    $\mA(b,0,\rho) = \vec{0}$
for all $\rho,b$. Now assume the claim is true for all times $s < t$. For the sake of contradiction, assume the claim does not hold for time $t$ and some realization $\rho$.

By definition of $\AMax$, there exists a number $M^*$ such that
$$
\sup_{b\in \dmn(M^*,y)}  \Norm{ \mA(b,t,\rho) } < \AMax(t,\rho,y) + \eps.
$$
For each ad $i$, define a new realization $\rho_i$ as follows: it coincides with $\rho$ before time $t$, only $i$ gets clicked at time $t$, and there are no clicks after $t$.

Fix bid vectors
    $b,b' \in \dmn(B_{\eps}(y),y)$.
Pick some bid vector
    $\bt\in \dmn(\widetilde{M},y)$,
where
    $$\widetilde{M} = \max(M^*,\,3y\,\Norm[\infty]{b+b'} ).$$
Let
    $M = \max(\bt)$.

\WMON for realization $\rho_i$, applied to bid vectors $b$ and $\bt$, states the following:
\begin{align}\label{eq:cl:ExPostWMON-rand-WMON}
&(\bt - b\,)^\dag\; \Delta_t(\rho_i)\;
    \left(\;  \mA(\bt,t,\rho) - \mA(b,t,\rho) \;\right) \\
&\quad\quad+
(\bt - b\,)^\dag\;
    \sum_{s=1}^{t-1} \Delta_s(\rho)\;
        \left( \mA(\bt,s,\rho) - \mA(b,s,\rho) \right)
\geq 0.
\end{align}

The first summand in~\eqref{eq:cl:ExPostWMON-rand-WMON} is simply $ (\bt_i-b_i)\,
        \left(\;  \mA_i(\bt,t,\rho) - \mA_i(b,t,\rho) \;\right)$.

By the induction hypothesis, for each time $s<t$ it holds that
$$ (\bt - b\,)^\dag\;
    \Delta_s(\rho)\;
        \left( \mA(\bt,s,\rho) - \mA(b,s,\rho) \right)
    \leq M\, \eps\,g(s,y)
$$
It follows that
$$
(\bt - b\,)^\dag\;
    \sum_{s=1}^{t-1} \Delta_s(\rho)\;
        \left( \mA(\bt,s,\rho) - \mA(b,s,\rho) \right)
        \leq M\, \eps\,\sum_{s=1}^{t-1} g(s,y)
        = M\, \eps\,G(t-1,y)
$$
Plugging this into~\eqref{eq:cl:ExPostWMON-rand-WMON}, we obtain
\begin{align*}
(\bt_i-b_i)\,
        \left(\;  \mA_i(\bt,t,\rho) - \mA_i(b,t,\rho) \;\right)
    \geq - M\, \eps\,G(t-1,y) \nonumber \\
\mA_i(\bt,t,\rho) - \mA_i(b,t,\rho)
    \geq - (\bt_i - b_i)^{-1} M \, \eps\,G(t-1,y).
\end{align*}
We have $\bt_i \geq M/y$ since $\bt$ is $y$-balanced.
Also $b_i \leq M/(3y)$ by our choice of $M$. Therefore
$\bt_i - b_i \geq \tfrac{2M}{3y}$ and
\begin{equation} \label{eq:cl:ExPostWMON-rand-2}
\mA_i(\bt,t,\rho) - \mA_i(b,t,\rho)
     \geq - \tfrac{3y}{2} \, \eps \, G(t-1,y).
\end{equation}

\xhdr{Case 1: $\bMax(t,\rho,\eps,y)<\infty$.}
Denote $X_i = \mA_i(b,t,\rho) $ and $X'_i = \mA_i(b',t,\rho) $. Let $\min(X_i,X'_i)$ be the coordinate-wise minimum of $X_i$ and $X'_i$; define $\max(X_i,X'_i)$ similarly.

In this notation, our goal is to bound $\Norm{X-X'}$ above by $\eps\,g(t,y)$.
Assume $\bt = M\,\vec{1}$ for some $M\geq \widetilde{M}$.
By
\eqref{eq:cl:ExPostWMON-rand-2}, noting that this argument applies to both $b$ and $b'$, we have:
$$\mA_i(\bt,t,\rho)
    \geq \max(X_i,X'_i) -  \tfrac{3y}{2}\, \eps\, G(t-1,y).
$$
Summing this over all ads: 
$$\Norm{ \mA(\bt,t,\rho) }
    \geq \Norm{ \max(X,X') } - \tfrac{3y}{2}\, \eps\,m\, G(t-1,y).
$$
Recall that
    $\Norm{ \mA(\bt,t,\rho) } \leq \AMax(t,\rho,y)+\eps$
by our choice of $M$. Therefore:
$$
\Norm{ \max(X,X') } \leq \AMax(t,\rho,y) +
    \eps\,(1+\tfrac{3y}{2}\,m\, G(t-1,y)).
$$

Note that
\begin{align*}
\Norm{X}+\Norm{X'} &= \Norm{\max(X,X')} + \Norm{\min(X,X')} \\
\Norm{X-X'} &= \Norm{\max(X,X')} - \Norm{\min(X,X')} \\
\Norm{X}+\Norm{X'} + \Norm{X-X'}
 &= 2 \Norm{\max(X,X')}
\end{align*}
Because $\bMax(t,\rho,\eps,y)<\infty$ and
$b,b' \in \dmn(B_\eps(y),y)$,
both $\Norm{X}$ and $\Norm{X'}$ are at least
    $\AMax(t,\rho,y) - \eps\, f(t,y)$.
Therefore:
\begin{align*}
 2\AMax(t,\rho) - 2\eps\, f(t,y) + \Norm{X-X'}
    &\leq 2\, \Norm{\max(X,X')} \\
    &\leq 2\AMax(t,\rho,y) + 2\,\eps\,(1+\tfrac{3y}{2}\, G(t-1,y)).
\end{align*}
It follows that
$$ \Norm{X-X'}
    \leq 2\eps \left(\; 1+f(t,y) + \tfrac{3y}{2}\,m\,G(t-1,y) \;\right)
    = \eps\,g(t,y).
$$
Thus, we have proved the induction step assuming $\bMax(t,\rho,\eps,y)$
is finite.

\xhdr{Case 2: $\bMax(t,\rho,\eps,y)= \infty$.}
This case is impossible: we will arrive at a contradiction.

By definition of $\AMax$, there exists a bid vector
    $b\in \dmn(B_\eps(y),y)$
such that
$$ \Norm{\mA(b,t,\rho)} > \AMax(t,\rho,y) - \tfrac12\, \eps\,f(t,y).
$$
Since $\bMax(t,\rho,\eps,y)= \infty$, we can pick
    $\bt \in \dmn(\widetilde{M},y)$
such that
$$\Norm{ \mA(\bt,t,\rho) }
    \leq \AMax(t,\rho,y) - \eps\,f(t,y)
    \leq  \Norm{ \mA(b,t,\rho)} - \tfrac12\,\eps\,f(t,y).
$$
It follows that
\begin{align*}
\sum_{i=1}^m \left[ \mA_i(b,t,\rho) - \mA_i(\bt,t,\rho) \right]
&\geq \tfrac12\,\eps\,f(t,y) \\
\exists i\quad
\mA_i(b,t,\rho) - \mA_i(\bt,t,\rho) &\geq \tfrac{1}{2m}\,\eps\,f(t,y).
\end{align*}
Using~\eqref{eq:cl:ExPostWMON-rand-2}, for this $i$ we have:
$$ \tfrac{1}{2m}\,\eps\,f(t,y)
    \leq \mA_i(b,t,\rho) - \mA_i(\bt,t,\rho)
    \leq \tfrac{3y}{2} \,\eps\,G(t-1,y).
$$
Thus, $f(t,y) \leq 3ym\,G(t-1,y)$, contradicting the definition of $f$.
\end{proof}

Using Claim~\ref{cl:ExPostWMON-rand}, it is now easy to prove
Theorem~\ref{thm:ExPostWMON}(c).
\begin{proof}[of Theorem~\ref{thm:ExPostWMON}(c)]
For any $\delta>0$, let
\begin{align*}
y &= 2m/\delta \\
\eps &= \tfrac{\delta}{2 m g(T,y)} \\
B &= B_{\eps}(y).
\end{align*}
In our proof we will considering applying $\mA$ to the bid
vector $b^0 = B \vec{1}$ as well as the vectors $b^j$
defined for $j=1,\ldots,m$ by changing the $j^{\mathrm{th}}$
of $b^0$ from $B$ to $yB$. The vectors $b^0,\ldots,b^m$
all belong to $\dmn(B,y)$.

Let $\rho$ be a realization such that $\rho(t,j)=1$ for
all $t,j$, i.e.\ every ad is always clicked.
Since $\mA$ can never allocate more than $T$ impressions, we have
$\sum_{t=1}^T \sum_{i=1}^m \mA_i(b^0,t,\rho) \leq T$.
Hence,  there is at least one $j \in [m]$ such that
\begin{equation} \label{eq:5.1.c.3}
\sum_{t=1}^T \mA_j(b^0,t,\rho) \leq T/m.
\end{equation}
Now, for every round $t$,
we have
\begin{equation} \label{eq:5.1.c.4}
\mA_j(b^j,t,\rho) - \mA_j(b^0,t,\rho) \leq
\Norm{ \mA(b^j,t,\rho) - \mA(b^0,t,\rho) } \leq
\eps \, g(t,y)
= \tfrac{\delta}{2m},
\end{equation}
where the second inequality follows from Claim~\ref{cl:ExPostWMON-rand}.
Summing~\eqref{eq:5.1.c.4} over $t=1,\ldots,T$ and combining
with~\eqref{eq:5.1.c.3},
we deduce that
\[
\sum_{t=1}^T \mA_j(b^j,t,\rho) \leq \left( 1 \tfrac{\delta}{2} \right)
\tfrac{T}{m}.
\]
The optimal allocation for bid vector $b^j$
assigns every impression to ad $j$, achieving
a total value of $yBT$. Instead, the allocation
computed by $\mA$ achieves a total value bounded above
by
       $\left(1 + \tfrac{\delta}{2} \right) \tfrac{yBT}{m} + BT$,
where the first term accounts for impressions allocated
to ad $j$ and the second term accounts for all other
impressions. We have
\[
\left(1 + \tfrac{\delta}{2} \right) \tfrac{yBT}{m} + BT =
\tfrac{yBT}{m} \cdot \left( 1 + \tfrac{\delta}{2} + \tfrac{m}{y} \right)
= yBT \cdot \tfrac{1+\delta}{m}.
\]
Since $\delta>0$ was an arbitrarily small positive constant,
we conclude that the worst-case approximation ratio of
$\mA$ is no better than $1/m$, which is trivially achieved by a
random allocation.
\end{proof}

\section{Multi-parameter MAB mechanisms:
A stochastic CMON allocation rule}
\label{sec:positive}


\newcommand{\ALG}{\mathtt{ALL}}

\newcommand{\Exploit}{\mathtt{Exploit}}
\newcommand{\Wper}{W_0} 

In this section we consider the problem of designing stochastically truthful multi-parameter MAB mechanisms.
As discussed in the introduction, the VCG mechanism cannot be used as it is informationally-infeasible. Additionally,
pricing based mechanisms do not seem to be feasible. The only other technique that is extensively exploited in the literature for multi-parameter domains is using \emph{maximal in distributional range} (\MIDR) allocation rules.
We formalize the limitations of a natural family of \MIDR allocation rules (in which the set of distributions the rule optimizes over is independent of the CTRs) in Section~\ref{subsec:MIDR}, showing that the performance of such rules is no better than randomly selecting an ad to present.
We next discuss some simple approaches to create truthful mechanisms: the first disregards the bids, and the second uses randomization to reduce the problem to a single parameter problem.

The first approach is \emph{bid-independent} allocation rules -- ones that do not depend on the bids.
Among those, we naturally focus on the allocation rule that achieves the best worst-case performance, that rule samples an ad independently and uniformly at random in each round; call it $\Rand$.

A slightly more sophisticated approach randomly reduces the problem to a single parameter problem as follows.
One ad is selected independently for each agent, uniformly at random from this agent's ads.
Then some truthful single-parameter mechanism $\mM$ is run on the selected ads. Call this mechanism $\MosheMech$. This mechanism is truthful (ex-post or stochastically, same as $\mM$) because
for each realization of the selection described above, it is simply a truthful single-parameter mechanism.
The performance of this mechanism is the same as the performance of the trivial $\Rand$ mechanism when there is only one agent.

These two na\"ive approaches have poor performance. For example, for a single agent none performs better than uniformly randomizing over the ads.
We call such a performance {\em trivial}.
This gives rise to the following major open problem.

\vspace{1mm}
{\bf Open Problem}: {\em Design a stochastically truthful mechanism for the multi-parameter MAB problem that achieves optimal approximation. }
\vspace{1mm}

A more modest goal is to design a stochastically truthful mechanism for the multi-parameter MAB problem that achieves {\em non-trivial} performance, even for some ``well-behaved'' subset of inputs.
Unfortunately, it seems that all standard tools fail to achieve even this modest goal.
Below we achieve this by designing a stochastically \CMON allocation rule and then applying the multi-parameter transformation from Section~\ref{sec:transform}. We interpret this result as an evidence that it is not completely hopeless to significantly improve over the trivial approaches.


\subsection{The stochastically CMON allocation rule}
\label{subsec:alloc}

We design a stochastically \CMON allocation rule $\ALG$ whose expected welfare exceeds  that of $\Rand$ on all problem instances with at least two agents, and that of $\MosheMech$  on an important family of problem instances which we characterize below. Structurally $\ALG$ depends on all submitted bids, is provably not \MIDR, and, unlike $\MosheMech$, does not proceed through an explicit reduction to a single-parameter allocation rule. Implementing $\ALG$ as a truthful, information-feasible mechanism requires the full power of our multi-parameter transformation.

All results in this section require all private values to be bounded from above by $1$. We will assume that without further notice.

\xhdr{Recap of notation.} The term ``expected welfare'' refers to expectation over the randomness in the allocation rule and the clicks (for a given vector of CTRs). Let $W(\Rand)$ denote the expected welfare of $\Rand$. Let $A_0=\{ 1 \LDOTS  m\}$ be the set of $m$ ads of all agents.
Recall that $v_j$, $b_j$ and $\mu_j$ be, resp., denote the private value, the submitted bid, and the CTR for ad $j$.
Note that the expected value from each time a given ad $j$ is displayed is $v_j \mu_j$.

\xhdr{Allocation rule $\ALG$ for $\geq 2$ agents.}
Assume there are at least two agents. Define the following allocation rule, call it $\ALG$. It consists of two phases: exploration and exploitation. Exploration lasts for $T_0$ rounds, where $T_0\geq 1$ is fixed and chosen in advance.
In each exploration round an ad is chosen uniformly at random among all ads. Let $n_j$ be the number of clicks for ad $j$ by the end of the exploration phase.
In each round of exploitation $\ALG$ does the following:
\begin{OneLiners}
\item[(L1)] pick each ad $j$ with probability $b_j\, n_j/T_0$,
where $b_j$ is the bid for ad $j$.
\item[(L2)] with the remaining probability pick an ad uniformly at random.
\end{OneLiners}
This completes the specification of $\ALG$. We note that even a single round of exploration suffices for our purposes. Using a small $T_0$ does not affect the expected performance, but results in a (very) high variance.

\xhdr{Discussion.} We design $\ALG$ to ensure that the allocation probabilities depend on CTRs and bids in a simple, linear way. Below we explain why this ``linear dependence" property is useful, and discuss some of the challenges in the analysis of $\ALG$.

Let the \emph{allocation-vector} be a vector $a\in \Re^m$ whose $j$-th component is the expected number of times ad $j$ is allocated by $\ALG$. For a given vector of CTRs, the \emph{allocation-range} is the set of all allocation-vectors that can be realized by $\ALG$. We conjecture that the allocation-range needs to depend on CTRs in order for an allocation rule to satisfy stochastic \CMON and be, in some sense, non-trivial. (In Section~\ref{subsec:MIDR}, we prove a version of this conjecture that is restricted to stochastically \MIDR allocation rules.) The ``linear dependence" property of $\ALG$ ensures that the allocation-range does depend on CTRs.

For example, consider an allocation rule which has an exploration phase of fixed duration, picks the best (estimated) ad based on the clicks received so far, and sticks with this ad from then on. This allocation rule that is ex-post truthful in the single-parameter setting, and is perhaps the most natural candidate for a reasonable, easy-to-analyze allocation rule for our setting. However, the allocation-range of this allocation rule does not depend on CTRs (because the set of possible options for exploitation is fixed: any one ad can be chosen).

Further, the proof technique that we use in the analysis of $\ALG$ essentially requires us, for every given agent, to solve a system of equations where the unknowns are this agent's bids and the parameters are the CTRs and the components of the allocation vector. The allocation probabilities in $\ALG$ are explicitly defined in terms of bids in order to enable us to solve this system of equations in a desirable way; this is another place where the ``linear dependence" property of $\ALG$ is helpful.


\OMIT{
Note that the system of equations which we need to solve for the bids is undetermined if there is only a single agent, essentially because then the allocation probabilities of all his ads in a given round must sum up to $1$. This is why $\ALG$ as stated is not stochastically \CMON for the single-agent case.}

The subtle point in our analysis of $\ALG$ -- or, it seems, in any analysis using the same proof technique -- is that one needs to ensure that the allocation vector is a maximizer of a certain expression, which requires us to prove the positive-definiteness of the corresponding Hessian matrix. The ``linear dependence" property of $\ALG$ enables us to argue about the Hessian matrix in a useful way.

As we discovered, the positive-definiteness of the Hessian should not be taken for granted: indeed, it fails for a number of otherwise promising allocation rules with better performance. We believe that further progress on stochastically \CMON allocation rules would require a more systematic understanding of how changes in the allocation rule propagate through the analysis and affect the Hessian matrix.

\xhdr{Guarantees for $\ALG$ for $\geq 2$ agents.}
A problem instance is called \emph{uniform} if the product $v_j \mu_j$ is the same for all $j$, and \emph{non-uniform} otherwise. Note that for uniform problem instances $\Rand$ is optimal, and in fact all allocation rules without skips have the same expected welfare,
and are all optimal. We will assume that all values-per-click are at most $1$, and that all CTRs are strictly positive.

\OMIT{
Note that instances on which $\MosheMech$ performs really poorly are those with one agent having all ads but one (belonging to a second agent), and only one of his ads has high value while all other ads have arbitrary low values. On the other hand, for such inputs $\ALG$ plays the best ad almost twice as often as $\MosheMech$.}

Note that instances on which $\Rand$ performs very poorly are those where for one ad $j$ the product $v_j \mu_j$ is large while for all other ads this product is very low. On the other hand, for such inputs $\ALG$ plays the best ad significantly more often.

We next present a parameter that aims to quantify the divergence of the instance from uniform and will be used to measure the performance of $\ALG$.
A problem instance is called $\sigma$-skewed, for some $\sigma\in [1,m]$, if it satisfies
\begin{align}\label{eq:positive-skewed}
    (M_2)^2 \geq \sigma (M_1)^2, \quad
    \text{where } M_q = \left( \tfrac{1}{m} \textstyle \sum_{j=1}^m\, (v_j\, \mu_j)^q \right)^{1/q}.
\end{align}
Note that problem instances can be $\sigma$-skewed for any given $\sigma\in [1,m]$. It is $1$-skewed for uniform problem instances, and $m$-skewed when only one ad is good while all other ads have value 0.

Let $\Wper(\ALG)$ be the expected per-round welfare for the exploitation phase of $\ALG$, and let $\Wper(\Rand)$ be the expected per-round welfare for $\Rand$. Note that $\Wper(\Rand)=M_1$.
The properties of $\ALG$ with at least two agents are captured by the next lemma (which is the main technical lemma in this section); its proof is deferred to Appendix~\ref{subsec:pf-lemma-positive}.

\begin{lemma}\label{lm:positive}
With at least two agents, allocation rule $\ALG$ satisfies the following:

\begin{itemize}
\item[(a)] If the CTRs for all ads are strictly positive then $\ALG$ satisfies stochastic \CMON.

\item[(b)]  For $\Wper(\ALG)$ and $\Wper(\Rand)$ as defined above it holds that
\begin{align*}
    \Wper(\ALG)-\Wper(\Rand) = M_2^2 - M_1^2, \quad
    \text{where } M_q = \left( \tfrac{1}{m} \textstyle \sum_{j=1}^m\, (b_j\, \mu_j)^q \right)^{1/q}.
\end{align*}
In particular, $W(\ALG)> W(\Rand)$ for all non-uniform problem instances.
\end{itemize}
\end{lemma}

The allocation rule $\ALG$ does not have the property that scaling all bids by a common factor scales the expected welfare by the same factor; therefore it is not \MIDR (see Section~\ref{subsec:MIDR} for the definition of \MIDR, as it applies to our setting).

\xhdr{Reduction to the single-agent case.}
For a single agent, we define our allocation rule $\ALG$ as follows: we simulate a run of $\ALG$ with a single round of exploration and two agents, where the second agent is a dummy agent with a single ad. The dummy agent submits a bid of zero for his ad, and we fix its CTR to $\tfrac12$ (any CTR works). This completes the specification of $\ALG$.

Denote the resulting two-agent allocation rule by $\ALG^*$. The single-agent allocation rule satisfies $\CMON$ because so does $\ALG^*$. Since the dummy agent does not contribute welfare (because of the zero bid), we have
    $\Wper(\ALG)=\Wper(\ALG^*)$.
Applying Lemma~\ref{lm:positive}(a) to $\ALG^*$, we see that
\begin{align}\label{eq:ALL-welfare-single}
\Wper(\ALG) = M^*_1+ (M^*_2)^2-(M^*_1)^2,
\text{ where }
M^*_q = \left( \tfrac{1}{m+1} \textstyle \sum_{j=1}^m\, (b_j\, \mu_j)^q \right)^{1/q}.
\end{align}

We summarize the useful properties of $\ALG$ in the following lemma:

\begin{lemma}\label{lm:positive-single}
Consider the case of a single agent; assume $\mu_j>0$ for all ads $j$. Then $\ALG$ satisfies stochastic \CMON, and its welfare in exploitation rounds satisfies \eqref{eq:ALL-welfare-single}. In the one exploration round, $\ALG$ obtains welfare $\tfrac{m}{m+1}\Wper(\Rand)$.

\end{lemma}

\xhdr{Main provable guarantee.} Let $\mM_\delta$ be the mechanism obtained by applying Theorem~\ref{thm:multiParam} to $\ALG$ with parameter $\delta\in (0,1)$. The main result of this section follows.

\begin{theorem}\label{thm:positive}
Consider a multi-parameter MAB domain with $v_j\leq 1$ and $\mu_j>0$ for every ad $j$. Then mechanism $\mM_\delta$ is stochastically truthful, for every $\delta\in (0,1)$.

Consider $\sigma$-skewed problem instances, 
and assume  $\max_{j\in A_0} v_j \mu_j>\eps>0$.
There exists $\delta\in (0,1)$ such that mechanism $\mM_\delta$ satisfies the following:

\begin{itemize}
\item[(a)] $W(\mM)> W(\Rand)$ on all problem instances with at least two agents, as long as
    $\sigma>1$.

\item[(b)] $W(\mM)> W(\Rand) = W(\MosheMech)$ on all problem instances with a single agent with $m$ ads,
 as long as
    $\sigma>1+\tfrac{m+1}{m\eps}+\tfrac{m+1}{\eps(T-1)}$.

\item[(c)] Suppose there exists an agent with $k>m/2$ ads; w.l.o.g. assume this is agent $1$. Then $W(\mM)> W(\MosheMech)$ on all problem instances such that
        $\sigma>1+ \tfrac{m(m-k)}{k\eps} $ 
    when for all agents $i>1$ all private values are $0$.
    \footnote{One can also derive a version of this result where the private values for all agents $i>1$ are smaller than $\delta$, for some $\delta\ll \eps$. We omit the easy  details.} \footnote{Note that the instances considered in this result are generalizing the instances we have discussed before. There are instances in which one agent have all but one ad, and only one of his ads has positive value, while all the rest of the ads (his and others) have value $0$.}
\end{itemize}
\end{theorem}

The theorem follows from Lemma~\ref{lm:positive}, Lemma~\ref{lm:positive-single} and
Theorem~\ref{thm:multiParam} via straightforward computations, some of which we omit from this version. Recall that for each $\delta>0$ we have
    $W(\mM_\delta) > (1-\delta)\, W(\ALG)$.

\begin{proof}[Theorem~\ref{thm:positive}(a)]
Assume $M_2 > (1+\eps) M_1$ and $\max_{j\in A_0} b_j \mu_j>\eps$ for some $\eps>0$. Then, using the notation of Lemma~\ref{lm:positive}(b), we have $M_1\geq \eps/m$, and therefore
$$ \Wper(\ALG)-\Wper(\Rand)
    \geq M_1^2\, ((1+\eps)^2-1) > M_1\, \tfrac{2\eps^2}{m}.
$$

\noindent Recall that $T_0$ is the duration of exploration in $\ALG$, and $T$ is the time horizon. Then:
\begin{align*}
W(\Rand)
    &= T\, \Wper(\Rand)  = T\, M_1\\
W(\ALG)
    &= T_0\,\Wper(\Rand) + (T-T_0)\, \Wper(\ALG) \\
     &= W(\Rand) + (T-T_0)\left( \Wper(\ALG)-\Wper(\Rand) \right) \\
    &> W(\Rand) + \gamma\, W(\Rand), \;
    \text{ where }  \gamma=  \tfrac{2\eps^2 (T-T_0)}{mT}  \\
W(\mM)
    &> (1-\eta)\, W(\ALG) > (1-\eta)(1+\gamma)\, W(\Rand).
\end{align*}
Thus, to ensure that $W(\mM)>W(\Rand)$, it suffices to take $\eta<1-\tfrac{1}{1+\gamma}$.
\end{proof}

\begin{proof}[Proof Sketch of Theorem~\ref{thm:positive}(bc)]
For part (b), recall that
    $\Wper(\Rand)=\Wper(\MosheMech) = \tfrac{1}{m}\sum_{j=1}^m\, b_j\, \mu_j$.
With a simple computation which we omit from this version, one derives that
    $W(\mA)> W(\Rand)$.
We prove $W(\mM)> W(\Rand)$ using a computation similar to the one in the proof of part (a), we omit the details.

For part (c), note that $\Wper(\Rand) = M_1$ and (under the assumptions in Theorem~\ref{thm:positive}(c)),
    $W(\MosheMech)\leq \tfrac{1}{k} M_1$.
Again, using a simple computation one can show that
    $W(\mA)>W(\MosheMech)$,
and then pick a sufficiently small $\delta$ as in the proof of part (a).
\end{proof}

\subsection{Proof of the main technical lemma (Lemma~\ref{lm:positive})}
\label{subsec:pf-lemma-positive}

\newcommand{\Eu}{E_{\mathtt{u}}}
\newcommand{\PP}{\mathcal{P}}

Let us set up some notation. Consider an exploitation round in the execution of $\ALG$. For each ad $j$, let $E_j$ be the event that ad $j$ is chosen in line (L1) of the algorithm's specification. Let $\Eu$ be the remaining event in line (L2) when the ad is chosen uniformly at random. Denote $x_j = \Pr[E_j]$, and note that for each ad $j$,
    $$x_j \triangleq \Pr[E_j] = b_j\, \E[n_j]/T_0 = \tfrac{1}{m}\, b_j \mu_j.$$

\begin{proof}[of Lemma~\ref{lm:positive}(b)]
Consider a round in the exploitation phase of $\ALG$. Partition this round into events $\PP = \{ E_1 \LDOTS E_m;\Eu \}$. For each event $E\in \PP$ in this partition, let $\Wper(E)$ be the expected per-round welfare of $\ALG$ from this event, so that
    $\Wper(\ALG) = \sum_{E\in \PP} \Wper(E)$.
Note that
    $\Wper(\Eu) = \Pr[\Eu]\, \Wper(\Rand)$.
Further,
    $\Wper(E_j) = b_j \mu_j \Pr[E_j] = m x_j^2$
for each ad $i$.

It is easy to see that
    $\Wper(\Rand) = \tfrac{1}{m} \sum_j b_j \mu_j = \sum_j x_j $.
It follows that
\begin{align*}
\Wper(\ALG) - \Wper(\Rand)
    &= \textstyle \sum_{E\in \PP} \Wper(E) - \Pr[E]\, \Wper(\Rand) \\
    &=  \textstyle \sum_j  \Wper(E_j) - \sum_j \Pr[E_j]\, \Wper(\Rand) \\
    &=  \textstyle \left( m \sum_j x_j^2\right) -
        \left(\sum_j x_j \right)^2 \\
    &= M_2^2 - M_1^2 . \qedhere
\end{align*}
\end{proof}

For Lemma~\ref{lm:positive}(a), we rely on the following characterization of \CMON from prior work:

\begin{lemma}\label{lm:affine-maximizer-general}
Consider a function $f: S \to \Re^k$, where $S\subset \Re^k$. Let $f(S)\subset \Re^k$ be the image of $f$. Then $f$ is \CMON if and only if it is an affine maximizer, i.e.
$$ f(x) = \argmax_{y\in f(S)} \left[ x\cdot  y - g(y) \right] \quad
\text{for some function $g: f(S) \to \Re$.}
$$
\end{lemma}

\begin{proof}[Proof of Lemma~\ref{lm:positive}(a)]
Assume that there are at least two agents, and all CTRs are strictly positive. Without loss of generality, let us focus on agent $1$. We will use the following notation. Let
    $A = \{ 1 \LDOTS k\}$
be the set of ads submitted by agent $1$. Here $k$ is the number of ads submitted by agent 1; note that $k<m$. Let
    $b = (b_1 \LDOTS b_k)$
be the vector of bids for agent $1$, where $b_j$ is the bid on ad $j$. Let $B = [0,1]^k$ be the set of all possible bid vectors for agent 1. Let $\mu = (\mu_1 \LDOTS \mu_k)$ be the vector of CTRs for agent 1.  We will use both $i$ and $j$ to index ads.

Throughout the proof, let us keep the bids of all other agents fixed. Let $C_{i,t}(b)$ be the expected number of clicks that ad $i$ receives in round $t$ of $\ALG$, given the bid vector $b$, where the expectation is taken over all realizations of the clicks and over the randomness in the algorithm.%
\footnote{Here it is more convenient to use a slightly different notation for click-vectors, compared to Section~\ref{sec:MAB-mech-problem}.}

 Let
    $\vec{C}_t(b) = \left( C_{1,t}(b) \LDOTS C_{k,t}(b) \right)$
be the round-$t$ vector over the ads of agent 1, and let
    $\vec{C}(b) = \sum_t \vec{C}_{i,t}(b)$
be the vector whose $i$-th component is the total expected number of clicks for ad $i$.

We need to prove that the function $\vec{C}:B\to \Re^k$ satisfies \CMON. It suffices to prove that \CMON is satisfied for each round $t$ separately, i.e. that it is satisfied for each function $\vec{C}_t$. This is obvious if $t$ is an exploration round. In the rest of the proof we fix $t$ to be an exploitation round.

By Lemma~\ref{lm:affine-maximizer-general}, it suffices to prove that $\vec{C}_t(b)$ is an affine maximizer, i.e. that
\begin{align}\label{eq:my-affine-t}
\vec{C}_t(b) = \argmax_{p\in \vec{C}_t(B)}\;
    \sum_{j\in A} b_i\, p_i - G(p,\mu)
\end{align}
for some function
    $G(p,\mu): \vec{C}_t(B) \times [0,1]^k \to \Re$,
where $\vec{C}_t(B)\subset [0,1]^k$ is the image of $\vec{C}_t$. Crucially, the function $G$ cannot depend on $b$.~
\footnote{Note that $G$ \emph{can} depend on the CTRs, even though the mechanism does not know them. This is because $G$ is only used for the analysis -- to prove \CMON, and it is not actually used in the mechanism.}

Denote $p^* = \vec{C}_t(b)$. If $p^*$ is an interior point of $\vec{C}_t(B)$ and function $G$ is differentiable, then \eqref{eq:my-affine-t} implies the following:
\begin{align}\label{eq:my-affine-t-derive}
\frac{\partial}{\partial p_i}\, G(p^*,\mu) = b_i \quad
    \text{for each ad $i$, bid vector $b$ and CTR vector $\mu$}.
\end{align}

\noindent We will construct a function $G(p,\mu)$ so that it satisfies~\eqref{eq:my-affine-t-derive}.

Here and on, $i\in A$ denotes an arbitrary ad of agent 1. Recall that
    $x_i = \Pr[E_i] = \tfrac{1}{m} b_i\mu_i$.
Thus:
\begin{align*}
\Pr[\Eu]
    &= \textstyle  1- \sum_{j\in A_0} \Pr[E_j] = 1- \sum_{j\in A_0} x_j \\
C_{i,t}(b)
    &= \textstyle  \mu_i\left( \Pr[E_i] + \tfrac{1}{m} \Pr[\Eu] \right)
    = \mu_i\left( x_i + \tfrac{1}{m} - \tfrac{1}{m} \sum_{j\in A_0} x_j \right).
\end{align*}
Recalling the notation $p^* = \vec{C}_t(b)$ and solving for $x_i$, we obtain
\begin{align*}
p^*_i/\mu_i &= x_i + \tfrac{1}{m} - \tfrac{1}{m} \textstyle \sum_{j \in A_0} x_j \\
\textstyle \sum_{j\in A} p^*_j / \mu_j
    &= \textstyle \sum_{j\in A} x_j + \tfrac{k}{m} -  \tfrac{k}{m} \sum_{j\in A_0} x_j \\
    &=  (\tfrac{k}{m} -Y) + (1-\tfrac{k}{m})  \textstyle \sum_{j\in A_0} x_j, \quad
        \text{ where } Y = \sum_{j\in A_0 \setminus A} x_j. \\
p^*_i/\mu_i &= \textstyle
    x_i - \alpha \sum_{j\in A} p^*_j / \mu_j +\beta.
\end{align*}
where
    $\alpha = \tfrac{1}{m-k}$
and
    $\beta = \tfrac{1}{m} - \alpha(Y-\tfrac{k}{m})$.
It follows that
\begin{align} \label{eq:positive-pStar}
b_i = \frac{m}{\mu_i}\; x_i
    = p^*_i\, \frac{m}{\mu_i^2} + \sum_{j\in A} p^*_j\, \frac{\alpha m}{\mu_i \mu_j} - \frac{\beta m}{\mu_i}.
\end{align}
Denote the RHS of~\eqref{eq:positive-pStar} by $f_i(p^*,\mu)$. We have proved that
    $b_i = f_i(p^*,\mu)$
for each ad $i$. Thus to obtain~\eqref{eq:my-affine-t-derive} it suffices to pick $G(p,\mu)$ so that it satisfies
\begin{align}\label{eq:positive-G-f}
    \frac{\partial}{\partial p_i} G(p,\mu) = f_i(p,\mu) \quad
    \text{for each $i\in A$}.
\end{align}
Integrating $f_i(p^*,\mu)$ over $p_i$, for each ad $i$, and combining the resulting expressions, we obtain
\begin{align}\label{eq:positive-G}
G(p,\mu) = -\sum_{i\in A}\, p_i\, \frac{m \beta}{\mu_i} + \frac{m}{2}\sum_{i\in A}\, p_i^2\, \frac{1+\alpha}{\mu_i^2}
    + \sum_{j\in A\setminus \{i\}}\, p_i p_j \, \frac{m\alpha}{\mu_i\mu_j}.
\end{align}
It is easy to check that this $G$ satisfies~\eqref{eq:positive-G-f}, which in turn implies~\eqref{eq:my-affine-t-derive}.~
\footnote{Write
    $f_i(p,\mu) = \phi_i + \sum_{j\in A} p_j\, \gamma_{ij}$
for some numbers $\phi_i$ and $\gamma_{ij}$. Then a function $G(p,\mu)$ satisfying~\eqref{eq:positive-G-f} exists if and only if
    $\gamma_{ij} = \gamma_{ji}$ for all $i\neq j$.}
It follows that for this $G$, $p=p^*$ is a critical point in~\eqref{eq:my-affine-t}.
From here on we will use the $G$ as defined in~\eqref{eq:positive-G}.

We claim that the critical point $p=p^*$ is in fact a local maximum in~\eqref{eq:my-affine-t}. Equivalently, we claim that $p=p^*$ is a local minimum of the function
$$\lambda(p) = G(p,\mu)-p\cdot b: \Re^k \to \Re.$$
For that, it suffices to prove that the Hessian matrix $H$ of $\lambda(\cdot)$,
defined by
$$ H_{ij}
    = \frac{\partial}{\partial p_i\, \partial p_j} \lambda(p)
    = \frac{\partial}{\partial p_i\, \partial p_j} G(p,\mu),
$$
is positive-definite for $p=p^*$ . Note that for any $p\in \Re^k$ it holds that
\begin{align}\label{eq:positive-Hessian}
 H_{ij} =
 \begin{cases}
        \tau \rho_i^2, & i=j, \\
        \rho_i \rho_j, & i\neq j,
    \end{cases}
\end{align}
where
    $\rho_i = \frac{\sqrt{\alpha m}}{\mu_i}$
for each $i\in A$, and
    $\tau = \tfrac{1+\alpha}{\alpha} = 1+m-k\geq 2$.
By Claim~\ref{cl:positive-definite}, such matrix is positive-definite.

To complete the proof, we will show that $p=p^*$ is the \emph{global} maximum in~\eqref{eq:my-affine-t} over all $p\in \Re^k$. For that, it suffices to prove that
that $p=p^*$ is the unique critical point over the entire $\Re^k$, i.e. the unique solution for the system
\begin{align}\label{eq:positive-critical}
 \frac{\partial}{\partial p_i}\, G(p,\mu) = b_i \quad
    \text{for each ad $i\in A$}.
\end{align}
Let us re-write this system using~\eqref{eq:positive-G-f}. (We find it convenient to use the notation $\tau$ and $\rho_i$, as in~\eqref{eq:positive-Hessian}.) Namely, for each $i\in A$ we have:
\begin{align*}
b_i + \tfrac{\beta m}{\mu_i}
    &= f_i(p,\mu) + \tfrac{\beta m}{\mu_i} \\
    &= p_i\, \left( \tau \rho_i^2 \right)
        + \textstyle \sum_{j\in A\setminus\{i\}} p_j\, \left( \rho_i \rho_j \right).
\end{align*}
It follows that the system in~\eqref{eq:positive-critical} is equivalent to
    $$ H\cdot p = w,$$
where the $k\times k$ matrix $H$ is defined by~\eqref{eq:positive-Hessian}, and the vector $w\in \Re^k$ is defined by
    $w_i = b_i + \tfrac{\beta m}{\mu_i} $
for all $i$. The matrix $H$ is non-singular (since it is positive-definite), so the system $H\cdot p=w$ has a unique solution $p$.
\end{proof}

\begin{claim}\label{cl:positive-definite}
Consider a $k\times k$ matrix $H$ given by~\eqref{eq:positive-Hessian}, where $\rho_1 \LDOTS \rho_k$ are arbitrary positive numbers. Assume
    $\tau\geq 1$.
Then $H$ is positive definite.
\end{claim}
\begin{proof}
We will use the \emph{Gram matrix} characterization of positive-definite matrices. Namely, to prove that $H$ is positive-definite, it suffices to construct finite-dimensional vectors $w_1 \LDOTS w_k$ such that
    $H_{ij} = w_i \cdot w_j$
for all $i,j$ and the vectors are linearly independent. Consider vectors
    $w_1 \LDOTS w_k \in \Re^{k+1}$ defined as follows:
$$ w_i(\ell) = \begin{cases}
    \sqrt{\tau-1}\, \rho_i,    &\ell=i,\; \ell\leq k\\
    0,    & \ell\neq i,\; \ell \leq k \\
    \rho_i,& \ell=k+1, \\
\end{cases}
$$
It is easy to see that these vectors satisfy the desired properties.
\end{proof}

\subsection{An impossibility result for stochastically MIDR allocation rules}
\label{subsec:MIDR}

\OMIT{let a \emph{click-vector} be a vector $c = (c_1 \LDOTS c_m) \in \Re^m$ such that $c_j$ is the expected number of clicks received by ad $j$. Similarly,}

Let us consider stochastically MIDR allocation rules for multi-parameter MAB mechanisms. We show that any such allocation rule (with a significant but reasonable restriction) is essentially trivial.

Let us formulate what it means for a given allocation rule $\mA$ to be stochastically \MIDR in our setting, in a specific way that is convenient for us to work with. For a given bid vector $b\in (0,\infty)^m$, and CTR vector $\mu \in [0,1]^m$, let the \emph{allocation-vector} be a vector $a=(a_1 \LDOTS a_m)$ such that $a_j$ is the expected number of times ad $j\in [m]$ is allocated by $\mA$. 
Note that the expected welfare corresponding to a given allocation vector $a$ is simply $\sum_j a_j b_j \mu_j$. Let
$ \mF_0= \{ a\in [0,T]^m:\,\textstyle \sum_j a_j \leq T \}
$
be the set of all feasible allocation-vectors. (The sum of the entries can be less than $T$ because skips are allowed.) Then $\mA$ is stochastically \MIDR if and only if for all bid vectors $b$ and all CTR vectors $\mu$ it holds that
\begin{align}\label{eq:MIDR-alloc}
W(\mA(b)) = \max_{a\in \mF} \textstyle \sum_j\, a_j b_j \mu_j
\end{align}
for some $\mF\subset \mF_0$ that does not depend on $b$, but can depend on $\mu$.

Note that~\eqref{eq:MIDR-alloc} does not immediately provide a stochastically truthful mechanism via VCG payments, because the computation of VCG payments is not immediately feasible without knowing the CTRs. In fact, \eqref{eq:MIDR-alloc} does not even provide an immediate way to compute the allocation (assuming $|\mF|\geq 2$), again because of the issue of not knowing the CTRs. This is in stark contrast with the prior work on \MIDR (which studied settings without the ''no-simulation'' constraint) where the \MIDR property immediately gave rise to a truthful mechanism via the VCG payment rule.

However, if an allocation rule satisfies~\eqref{eq:MIDR-alloc} then a truthful mechanism can be obtained, with an arbitrarily small loss in welfare, via the transformation in~\cite{SingleCall-ec12}.

We consider a restricted version of~\eqref{eq:MIDR-alloc} where the range $\mF$ cannot depend on the CTRs (we will call such range $\mF$ \emph{CTR-independent}). We prove that any such allocation rule is welfare-equivalent to a time-invariant allocation rule. Here an allocation rule is called \emph{time-invariant} if in each round, it picks an ad independently from the same distribution over ads (this distribution may depend on the bids). Note that time-invariant allocation rules ignore the feedback that they receive (i.e., the clicks), and thus cannot adjust to the CTRs.

\begin{lemma}\label{lm:MIDR-impossibility}
Consider a multi-parameter MAB domain. Let $\mA$ be a stochastically \MIDR allocation rule with CTR-invariant range. For each bid vector $b$ there exists an allocation-vector $a=a(b)\in \mF_0$ such that
    $W(\mA(b)) = \textstyle \sum_j a_j b_j \mu_j$
for all CTR vectors $\mu$.
So $\mA$ is welfare-equivalent to a time-invariant allocation rule (where, letting $T$ be the time horizon, each ad $j$ is chosen with probability $a_j(b)/T$).
The approximation ratio of $\mA$ (compared to the welfare of the best ad) is at least $m$ on some problem instances.
\end{lemma}

\OMIT{
\begin{corollary}
In the setting of Lemma~\ref{lm:MIDR-impossibility}, the approximation ratio of $\mA$ (compared to the welfare of the best ad) is at least $m$ on some problem instances.
\end{corollary}
} 

\begin{proof}
Let us fix the bid vector $b$ and consider both sides of~\eqref{eq:MIDR-alloc} as functions of $\mu$. First, we note that the expected welfare $W(\mA(b))$ is a finite-degree polynomial in variables $\mu_1 \LDOTS \mu_m$.~
\footnote{Namely, the degree is at most $T$, the time horizon.}
This is because, letting $\mA_j(b,\rho,t)$ be the probability that ad $j$ is displayed at round $j$ given click-realization $\rho$, it holds that
\begin{align}\label{MIDR:W-polynomial}
W(\mA(b)) = \textstyle \,\sum_\rho \Pr[\rho]\;
    \sum_{j,t} \rho(t,j)\, \mA_j(b,\rho,t).
\end{align}
Here the outer sum is over all click-realizations $\rho$, and the inner sum is over all rounds $t$ and all ads $j$. $\Pr[\rho]$ is the probability that $\rho$ is realized for the given CTR vector. \eqref{MIDR:W-polynomial} is a polynomial in the CTRs because for each click-realization $\rho$, the inner sum is a fixed number, and $\Pr[\rho]$ is a polynomial in the CTRs of degree $T$.

Let us re-write~\eqref{eq:MIDR-alloc} as follows:
\begin{align}\label{eq:MIDR-alloc-rewrite}
W(\mA(b)) = \max_{\beta\in \mF_b} \textstyle \beta \cdot \mu,
\text{ where }
\mF_b = \{ \beta\in \Re^m: \beta_j = a_j b_j \text{ for each $j$}, a\in \mF  \}.
\end{align}

Since $\mF_b$ is fixed, the right-hand side of~\eqref{eq:MIDR-alloc-rewrite} is uniquely determined by $\mu$, denote it $W(\mu)$. Note that for each $\beta\in \mF_b$, it holds that
$$ W(\mu) = \beta\cdot\mu
    \;\text{if and only if}\;
(\beta-\beta')\cdot \mu\geq 0 \text{ for all } \beta'\in \mF_b.
$$
For each $\beta \in \Re^k$, consider the half-space
    $H_\beta = \{ \mu \in [0,1]^m:\, \beta\cdot \mu\geq 0\}$.
Then
$$ W(\mu) = \beta\cdot \mu
    \;\text{if and only if}\;
    \mu \in S_\beta,
    \;\text{where}\;
    S_\beta = \bigcap_{\beta'\in \mF_b} H_{\beta-\beta'}.
$$
Note that $S_\beta$ is a convex set, as an intersection of convex sets. Moreover, all half-spaces in the intersection contain the $0$-vector, and hence so does $S_\beta$. Therefore if $W(\mu) = \beta\mu$ for some $\mu \neq 0$ and $\beta\in \mF_b$ then by convexity for any $z\in [0,1]$ it holds that $z\mu \in S_\beta$, and therefore
    $W(z \mu) = z\,(\beta\cdot\mu) = z\,W(\mu)$.
We have proved the following:
\begin{align}\label{eq:MIDR-homogen}
W(z \mu) = z\,W(\mu)
\text{ for every $z\in [0,1]$ and $\mu\in[0,1]^m$}.
\end{align}

Now recall that $W(\mu)$ is a finite-degree polynomial in $\mu$. A known fact about multi-variate polynomials is that any finite-degree polynomial in $\mu$ which satisfies \eqref{eq:MIDR-homogen} is in fact of the form $W(\mu) = \gamma\cdot \mu$ for some $\gamma\in \Re^m$.

Now, let $A=\{j: b_j>0\}$ be the set of ads with non-zero bids. Define a vector $a\in \Re^m$ by $a_j = \gamma_j/b_j$ for each ad $j\in A$, and $a_j=0$ otherwise. To complete the proof, it remains to show that, letting $T$ be the time horizon, $a/T$ is a valid distribution over the ads (assuming skips are allowed). That is, we need to show that $a_j\geq 0$ and
    $\sum_{j\in A} a_j \leq T$.
We use the fact that for any allocation rule, the expected welfare is at least $0$ and at most that of always playing the best ad:
\begin{align}\label{eq:MIDR-last-step}
    W(\mu) = \textstyle \sum_{j\in A} a_j b_j \mu_j \in [0, T\,\max_j b_j \mu_j].
\end{align}
Applying~\eqref{eq:MIDR-last-step} with $\mu$ being the unit vector in the direction $j\in A$, it follows that $W(\mu) = a_j b_j \geq 0$, so $a_j\geq 0$.
Now, let
    $B = (\max_{j\in A} b_j^{-1})^{-1}$
and define a CTR vector $\mu$ by $\mu_j = B/b_j$ for $j\in A$ and $\mu_j=0$ otherwise.%
\footnote{The $B$ is needed to make sure that $\mu_j\leq 1$.}
 Plugging this $\mu$ into~\eqref{eq:MIDR-last-step}, we obtain
    $W(\mu) = \sum_{j\in A} B a_j  \leq B T  $,
which implies $\sum_{j\in A} a_j \leq T$, completing the proof.
\end{proof}

\begin{small}
\bibliographystyle{plainnat}
\bibliography{bib-abbrv,bib-slivkins,bib-AGT,bib-bandits}
\end{small}


\end{document}